%% file: JnlVer1.tex
\documentclass[report]{llncs}

\usepackage{pslatex,amssymb,pgf,tikz,lscape,amsmath,cite, mathrsfs,caption}
\usetikzlibrary{arrows,automata,positioning}

\input{macros}

\begin{document}
\title{Deterministic Logics for $UL$}
\author{Paritosh K.~Pandya and Simoni S.~Shah}
\institute{Tata Institute of Fundamental Research,
Colaba, Mumbai \textit{400005}, India}
\pagestyle{empty} 

\maketitle
\begin{abstract}
The class of Unambiguous Star-Free Regular Languages ($UL$) was defined by Schutzenberger as the class of languages defined by Unambiguous Polynomials. $UL$ has been variously characterized (over finite words) by logics such as \tlxy, \uitl, \utl, \fotwoless, the variety $DA$ of monoids, as well as partially-ordered two-way DFA (\potdfa). We revisit this language class with emphasis on notion of unambiguity and develop on the concept of Deterministic Logics for UL. The formulas of deterministic logics uniquely parse a word in order to evaluate satisfaction. We show that several deterministic logics robustly characterize $UL$. Moreover, we derive constructive reductions from these logics to the \potdfa\/ automata. These reductions also allow us to show \np-complete satisfaction complexity for the deterministic logics considered.

Logics such as \utl, \fotwoless\/ are not deterministic and have been shown to characterize $UL$ using algebraic methods. However there has been no known constructive reduction from these logics to \potdfa. We use deterministic logics to bridge this gap. The language-equivalent \potdfa\/ for a given \utl\/ formula is constructed and we analyze its size relative to the size of the \utl\/ formula. This is an efficient reduction which gives an alternate proof to \np-complete satisfiability complexity of \utl\/ formulas.
\end{abstract}

\section{Introduction}
Unambiguous star-free regular languages ($UL$) was a language class first studied by Sch\"utzenberger \cite{Sch76}. He gave an algebraic characterization for $UL$ using the monoid variety $DA$. Since then, several diverse and unexpected characterizations have emerged for this language class: $\Delta_2[<]$ in the quantifier-alternation hierarchy of first-order definable languages \cite{PW97}, the two variable fragment $FO^2[<]$ \cite{TW98} (without any restriction on quantifier alternation), and Unary Temporal Logic \utl\/ \cite{EVW02} are some of the logical characterizations that are well known. Investigating the automata for $UL$, Schwentik, Therien and Volmer \cite{STV01}  defined Partially Ordered 2-Way Deterministic Automata ($\potdfa$) and showed that these exactly recognize the language
class $UL$. Recently, there have been additional characterizations of $UL$ using deterministic logics $\uitl$ \cite{LPS08} as well as $\tlxy$ \cite{DK07}. A survey paper \cite{DGK08} describes this language class and its characterizations.

A monomial over an alphabet $\Sigma$ is a regular expression of the form $A_0^* a_1 \cdots a_{n-1}A_n^*$, where $A_i\subseteq\Sigma$ and $a_i\in\Sigma$. By definition, $UL$ is the subclass of star-free regular languages which may be expressed as a finite disjoint union of unambiguous monomials: every word that belongs to the language, may be \emph{unambiguously} parsed so as to match a monomial. The uniqueness with which these monomials parse any word is the characteristic property of this language class. 
We explore a similar phenomenon in logics  by introducing the notion of \emph{Deterministic Temporal Logics for $UL$}.

Given a modality $\mathscr M$ of a temporal logic that is interpreted over a word model, the \emph{accessibility relation} of $\mathscr M$ is a relation which maps every position in the word with the set of positions that are accessible by $\mathscr M$. In case of interval temporal logics, the relation is over intervals instead of positions in the word model. The modality is \emph{deterministic} if its accessibility relation is a (partial) function. A logic is said to be deterministic if all its modalities are deterministic. Hence, deterministic logics over words have 
the property of \emph{Unique Parsability}.
\begin{definition}[Unique Parsability]
In the evaluation of a temporal logic formula over a given word, every subformula has a unique position (or interval) in the word at which it must be evaluated. This position is determined by the context of the subformula. 
\end{definition}

In this paper we relate various deterministic temporal logics with diverse deterministic temporal modalities and investigate their properties. We give constructive reductions between them (as depicted in Figure \ref{fig:unamb}) and also to the \potdfa\/ automata. Hence, we are able to infer their expressive equivalence with the language class $UL$. Moreover, the automaton connection allows us to establish their \np-complete satisfiability for \emph{all} the deterministic logics that are considered.
\begin{enumerate}
\item[(i)]Deterministic Until-Since Logic- \duds: \\
Let $A$ be any subset of the alphabet and $b$ be any letter from the alphabet. The "deterministic half until" modality $A \detu_b \phi$ holds if at the first occurrence of
$b$ in (strict) future $\phi$ holds and all intermediate letters are in $A$. The past operator $A \dets_b \phi$ is symmetric. Since the modalities are deterministic, the formulas posses the property of unique parsability. This logic admits a straightforward encoding of \potdfa.

\item[(ii)]Unambiguous Interval Temporal Logic with Expanding Modalities - \uitlpm:\\
This is an interval temporal logic with deterministic chop modalities $\firsta$ and $\lasta$
which chop an interval into two at the first or last occurrence of letter $a$. These modalities were introduced in \cite{LPS08} as logic $\uitl$. Here,  we enrich \uitl\/ with the expanding $\firstp{a}$ and $\lastm{a}$ chop modalities that extend an interval beyond the interval boundaries in the forward and the backward directions to the next or the previous occurrence of $a$. We call this logic \uitlpm. 

\item[(iii)]Deterministic Temporal Logic of Rankers -\tlxy:\\
Modality $X_a \phi$ (or $Y_a \phi$)  accesses the position of the next (or the last) occurrence of
letter $a$ where  $\phi$ must hold. The temporal logic with these modalities was investigated in \cite{DK07}. The authors showed that the deterministic temporal logic $\tlxy$ which closes the rankers of \cite{WI07,STV01} under boolean operations, characterizes $UL$ (their work was in the setting of infinite words). We identify \tlxy\/ as a deterministic logic and use its property of unique parsability to give an efficient reduction from  formulas to  \potdfa.

\item[(iv)]Recursive Deterministic Temporal Logic - \tlrecr:\\
This logic has the recursive modalities $X_\phi$ and $Y_\phi$. These modalities deterministically access (respectively) the next and previous positions where the formula $\phi$ holds. $\phi$ in turn, is a \tlrecr\/ formula. An attempt to ``flatten'' the \tlrecr\/ formulas by a reduction to \tlxy\/ formulas seems non-trivial. However we observe another important property of rankers namely \emph{convexity}. This property holds true even in the case of recursive rankers. Using this property, we  give a polynomial time reduction from \tlrecr\/ to the non-deterministic $\utl$.  
\end{enumerate}

The above logics share some common properties: all their modalities are \emph{deterministic} and they possess the property of unique parsability. This is the key property which brings out the ``unambiguity'' of the language class. The above logics are also symmetric- in the sense that they possess both \emph{future} and \emph{past} type of modalities. This property corresponds to the two-way nature of the \potdfa\/ automata and we are able to show constructive equivalences between the logics and \potdfa. 

\cite{DKL10} showed an important property of the logic \tlxy\/ namely \emph{ranker directionality}: Given a ranker $r$ there exist \tlxy\/ formulas which determine the relative positioning of any position in the word with respect to the position at which $r$ accepts. This property has proved to be crucial in the translation from various logics of $UL$ to \tlxy.

The prominent logical characterizations of $UL$ have primarily been non-deterministic, such as the fragments $\Delta_2[<]$ and $FO^2[<]$ of first-order definable languages and as Unary Temporal Logic $\utl$. While these  logics are expressively equivalent to  Partially ordered 2-Way DFAs ($\potdfa$), no explicit reductions from these logics to $\potdfa$ were known. Neither the complexities of the formula automaton construction nor the bounds on the size of equivalent automata were worked out.  We give an effective language preserving translation from the non-deterministic
logic \utl\/ to the deterministic logic $\tlxy$. This completes the missing link
in effective reduction from logics $\utl$ and $FO^2[<]$ for $UL$ to their language equivalent $\potdfa$ automata. (See figure \ref{fig:unamb})
The translation is complex and its formulation involves ranker directionality along with following key observation which relates unary \emph{future} and \emph{past} modalities to the deterministic \emph{first} and \emph{last} modalities: 
\begin{quote}
In order to evaluate the truth of a \utl\/ formula $\fut(\phi)$ or $\past(\phi)$ at any position $i$ in a word $w$, it is sufficient to determine the ordering of $i$ relative to the first and last positions in $w$ at which its immediate modal subformula $\phi$ holds.
\end{quote}
The logic $\utl$ was shown
to have \np-complete satisfiability, originally by Etessami, Vardi and Wilke  \cite{EVW02}, by exploiting its small-model property. Our translation from \utl\/ to \tlxy\/ and hence \potdfa, gives an alternative ``automata-theoretic'' proof for the same and allows us to analyze the structure and size of the resulting language-equivalent automaton.

\begin{figure}
\begin{tikzpicture}{scale=0.5,transform shape}
\draw (15,6) node (B) [rectangle, draw] {$\fotwoless$};
\draw (15,10) node (C) [rectangle, draw] {\utl};
\draw (10,3) node (D) [rectangle, draw] {\duds};
\draw (10,6) node (E) [rectangle,draw] {\uitlpm};
\draw (10,9) node (F) [rectangle,draw] {\tlxy};
\draw (10,12) node (G) [rectangle,draw] {\tlrecr};
\draw (5,7) node (H) [rectangle, draw] {\potdfa};

\draw (B)[arrows= -triangle 45] .. controls (16,8) .. node[anchor=west]{$\mathcal O(2^n)$} (C);
\draw (C)[arrows= -triangle 45].. controls (14,8).. node[anchor=east] {$\mathcal O(n)$} (B);
\draw (E) [arrows= -triangle 45] -- node[anchor=east]{$\mathcal O(n^2)$}(F);
\draw (D) [arrows= -triangle 45] -- node[anchor=west]{$\mathcal O(n^2)$} (E);
\draw (H) [arrows= -triangle 45] -- node[anchor=east]{$\mathcal O(2^n)$} (D);
\draw (C) [arrows= -triangle 45] --node[anchor=north]{$\mathcal O(2^n)$} (F);
\draw (G) [arrows= -triangle 45] -- node[anchor=north]{$\mathcal O(n)$}(C);
\draw (F) [arrows= -triangle 45] -- node[anchor=north]{$\mathcal O(n^2)$}(H); 
\draw[dotted] (12,2) -- (12,14);
\draw (10,13.5) node{Deterministic};
\draw (14,13.5) node {Non-deterministic};
\draw (C)  [arrows= -triangle 45] .. controls (13,12) .. node[anchor=south]{$\mathcal O(n)$} (G);
\end{tikzpicture}
\caption{Unambiguous Languages and its equivalent characterizations: Arrows indicate the size blow-up in the effective reduction in the corresponding direction}
\label{fig:unamb}
\end{figure}
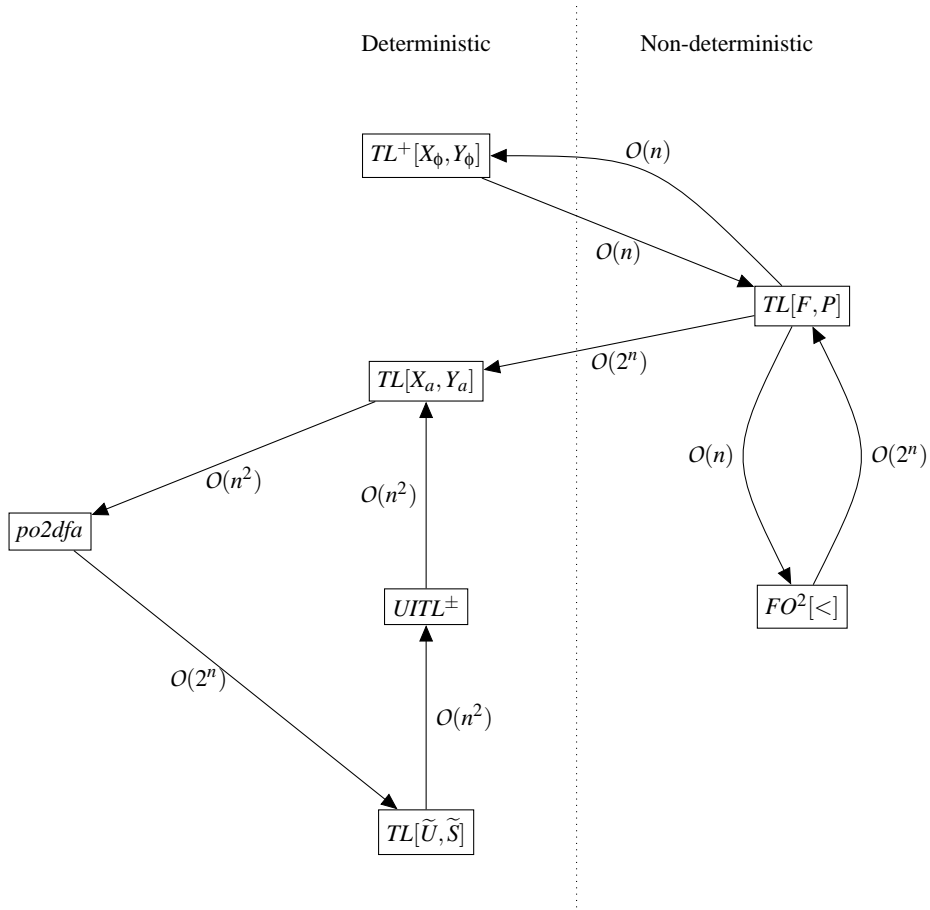

This paper is organized as follows.

\section{\potdfa: An Automaton characterization for UL}
\label{sec:potdfa}
Partially ordered two-way DFA were introduced by Schwentick, Th\'erien
and Vollmer \cite{STV01} where they showed that it is characterized by $DA$. As the name suggests, \potdfa\/ are two-way automata, so that the head of the automaton may move in either direction (one step to the left or right) in every transition. Also, the only loops in the transition graph of the automaton are self-loops on states. This naturally defines a partial-order on the set of states. Lastly, the automaton is deterministic- so that there is exactly one possible transition from any configuration of the automaton.\\
\\
Consider a finite alphabet $\Sigma$. Given $w \in \Sigma^*$, the two way automaton actually scans the string  $w'=\lend w\rend$ with end-markers $\lend$ and $\rend$ placed at positions 0 and  $\#w+1$ respectively. Let $\Sigma' = \Sigma \cup \{\lend,\rend\}$ include the two endmarkers.

\begin{definition}[\potdfa]
A \potdfa\/ over $\Sigma$ is a tuple $M = (Q,\leq,\delta,s,t,r)$ where $(Q,\leq)$ is a poset of states such that $r,t$ are the only minimal elements. $s$ is the initial state, $t$ is the accept state and $r$ is the rejecting state. The set $Q \setminus \{t,r\}$ is partitioned into $Q_L$ and $Q_R$ (the states reached from the left and the right respectively). 
  $\delta: ((Q_L \cup Q_R)\times \Sigma)\to Q)
          \cup((Q_L\times\{\rend\})\to Q\setminus Q_R)
	  \cup((Q_R\times\{\lend\})\to Q\setminus Q_L)$
is a progress-transition function satisfying $\delta(q,a) < q$. Hence it defines the progress transitions of the automaton. In order to
make the automaton \textit{``complete''}, every state $q$ in $Q \setminus
\{t,r\}$ has a default \textit{else} (self-loop) transition which is taken on all
letters $b$ for which no progress transition $\delta(q,b)$ is defined. 
Hence, the transition function $\delta$ specifies all the \emph{progress} transitions of the automaton, and a default self-loop (\emph{else}) transition is takes place otherwise. Note that there are no progress or \emph{else} transitions for the terminal states ($r$ and $t$). 
\end{definition}

\paragraph{Direction of head movement on a transition\\}
The direction in which the head moves at the end of a transition, depends on whether the target state of the transition is a $Q_L$ state, or a $Q_R$ state. 
$Q_L$ is the set of states that are \emph{``entered from the left''} and $Q_R$ are the states that are \emph{``entered from the right''}; i.e. 
if the automaton is in a state $q$, reading a symbol $a$, it enters a state $q'=\delta(q,a)$, then it moves its head to the right if $q'\in Q_L$, left if $q'\in Q_R$, and stays in the same position if $q'\in \{t,r\}$. The same rule applies to the self loop \emph{else} transitions also: on \emph{else} transitions of $Q_L$ states, the head moves to the right, and on \emph{else} transitions of $Q_R$ states, the head moves to the left.

\paragraph{Transitions on end-markers\\}
The transition function is designed to ensure that the automaton does not "fall off" either end of the input. Hence, for all $q\in Q\setminus \{t,r\}$, there are transitions $\delta(q,\lend)\in Q_L\cup\{t,r\}$ and $\delta(q,\rend)\in Q_R\cup\{t,r\}$.

\paragraph{Run of a \potdfa\\}
A po2dfa $M$ running over word $w$ is said to be in a configuration $(q,p)$ if it is in a state $q$ and head reading the position $p$ in  word. Let $Def(q)\subseteq \Sigma$ be the subset of letters on which no progress transition from $q$ is defined. Hence, the automaton takes the default \emph{else} transition on exactly the letters from $Def(q)$.  The run of a po2dfa $M$ on an input word $w$ starting with input head position $p_0$  is a sequence  $(q_0,p_0),(q_1,p_1),...(q_f,p_f)$ of configurations such that:
\begin{itemize}
\item $q_0=s$ and $q_f\in \{t,r\}$, 
\item For all $i  (1\leq i<f)$, if $w(p_i)\in Def(q_i)$ then 
\begin{itemize}
\item $q_{i+1}=q_i$ and
\item $p_{i+1}=p_i+1$ if $q_{i}\in Q_L$ and $p_{i+1}=p_i-1$ if $q_i\in Q_R$.
\end{itemize}
Otherwise, if $\delta(q_i,w(p_i))=(q')$ then
\begin{itemize}
\item $q_{i+1}=q'$ and
\item $p_{i+1}=p_i+1$ if $q_{i+1}\in Q_L$,\\
 $p_{i+1}=p_i-1$ if $q_{i+1}\in Q_R$  and \\
$p_{i+1}=p_i$ if $q_{i+1}\in\{t,r\}$.
\end{itemize}
\end{itemize}
In general, we abbreviate the run of an automaton $M$ starting from a position $p_0$ in a word $w$ by writing $M(w,p_0)=(q_f,p_f)$.
The run is \textit{accepting} if $q_f=t$; \textit{rejecting} if $q_f=r$. 
The automaton $M$ is said to be \textit{start-free} if for any $w$, and $\forall p_1,p_2\in dom(w)$,
$M(w,p_1) = (q_f,p_f)$ if and only if $M(w,p_2)=(q_f,p_f)$. 

The language $\mathcal L(M)$ of a \potdfa\/ $M$ is the set of all words $w$ such that $M(w,1)=(t,i)$ (for some $i\in dom(w')$).

\begin{remark}
We shall represent \potdfa\/ using their transition graphs such that all $q\in Q_L$ are marked with a ``$\rightarrow$'' and all $q\in Q_R$ are marked with a ``$\leftarrow$''. 
\end{remark}

\begin{example}
\label{exam:potdfa}
The \potdfa\/ $\autm$ is given in figure \ref{fig:expotdfa}. $\autm$ accepts all such words over $\{a,b,c,d\}^*$, which has its last $a$ at some position (say $i$), and some position (say $j>i$) has the first $d$ after $i$ and all intermediate positions between $i$ and $j$ do not have a $b$. Observe that the automaton rejects iff:
\begin{itemize}
\item There is no $a$ in the word
\item There is no $d$ after the last $a$ in the word
\item There is a $b$ between the last $a$ and the subsequent $d$ after it.
\end{itemize}
The language accepted by $\autm$, may be given by the regular expression $\Sigma^* a c^* d \{b,c,d\}^*$.
\end{example}

\begin{center}
\begin{figure}
\begin{tikzpicture}
[scale=0.8,transform shape, node distance = 2.5cm, initial text=]
 \tikzstyle{every state}=[thick]

\node[state,initial](S) {$\overrightarrow{s}$};
\node[state](A)[right of= S] {$\leftarrow$};
\node[state](B)[right of=A] {$\rightarrow$};
\node[state](T)[right of= B] {$t$};
\node[state](R)[below of = T] {$r$};
\path[->]
(S) edge [above] node {$\rend$} (A)
(A) edge [above] node {$a$} (B)
(A) edge [above] node {$\lend$} (R)
(B) edge [above] node {$b,\rend$} (R)
(B) edge [above] node {$d$} (T)
;
\end{tikzpicture}
\caption{Example \potdfa $\autm$}
\label{fig:expotdfa}
\end{figure}
\end{center}

\subsection{Constructions on \potdfa}
\label{sec:ete}

For the description of \POTDFAO\/ we shall use \name{Extended Turtle Expressions} (\cite{LPS08}), which are extensions of the turtle 
programs introduced  by Schwentick, Th\'erien and Vollmer \cite{STV01}. 
The syntax of \ete\/ follows and we explain its semantics below.  
Let 
$A,B$ range over subsets of $\Sigma'$.
\[
E ::= \begin{array}[t]{l}
        Acc ~\mid~ Rej ~\mid~  \funit{A} ~\mid~ \bunit{A} ~\mid~
        \fscan{A}{B} ~\mid~ \bscan{A}{B} ~\mid~ E_1?E_2,E_3
      \end{array}
\]

Automaton $Acc$ accepts immediately without moving the head. Similarly,
$Rej$ rejects immediately. 
$\fscan{A}{B}$ accepts at the next occurrence of a letter from $B$ 
strictly to the right, maintaining the constraint that the intervening 
letters are from $A \setminus B$. 
If no such occurrence exists the automaton rejects at the right end-marker or 
if a letter outside $A$
intervenes, the automaton rejects at its position. 
Automaton $\funit{A}$ accepts one position to the right if the current 
letter is from $A$, else rejects at the current position. 
$\bscan{A}{B}$ and $\bunit{A}$ are symmetric in the leftward direction. 
The conditional construct $E_1?E_2,E_3$ first 
executes $E_1$ on $w$. On its accepting $w$ at position $j$ it continues with
execution of $E_2$ from $j$. On $E_1$ rejecting $w$ at position $j$ it continues
with $E_3$ from position $j$. 

Here are some abbreviations which illustrate the power of the notation: 
$E_1;E_2 = E_1?E_2,Rej,~~~
\lnot E_1 = E_1?Rej,Acc$. Moreover, if $E_2$ is start-free then
$E_1 \lor E_2 = E_1?Acc,E_2$ and $E_1 \land E_2 = E_1?E2,Rej$.
Notice that automata for these expressions are start-free if $E_1$ is 
start-free.  We will use $\fscan{A}{a}$ for $\fscan{A}{\{a\}}$,
$\fsearch{a}$ for $(\fscan{\Sigma'}{a})$ and 
$\fstep$ for $(\funit{\Sigma'})$. 
Similarly define $\bsearch{a}$ and $\bstep$.

\begin{proposition}
\begin{itemize}
\item Given an \ETE\/ $E$ we can construct a \POTDFAO\/ accepting 
the same language with number of states linear in $|E|$.
\item Given a \potdfa\/ $\autm$ we may construct a language-equivalent \ete\/ whose size is linear in the size of $\autm$. 
\end{itemize}
\end{proposition}

\subsection{Properties of \potdfa}
The following properties of \potdfa\/ are useful. See \cite{LPS08} for details.
\begin{itemize}
\item \emph{Boolean Closure}: Boolean operations on \potdfa\/ may be achieved with linear blow-up in the size of the automata.
\item \emph{Small Model}: Given a \potdfa\/ $M$ with $n$ number of states, if $\mathcal L(M) \neq \emptyset$, then there exists a word $w\in\mathcal L(M)$ such that length of $w$ is linear in $n$.
\item \emph{Membership Checking}: Given a \potdfa\/ $M$ with $n$ number of states and a word $w$ of length $l$, the membership of $w$ in $\mathcal L(M)$ may be checked in time $\mathcal O(nl)$.
\item \emph{Language Non-Emptiness}: The non-emptiness of the language of a \potdfa\/ may be decided with \np-complete complexity.
\item \emph{Language Inclusion}: The language inclusion problem of \potdfa\/ is \conp-complete.
\end{itemize}

\section{\tlxy}
In \cite{DK07} the authors showed that the deterministic temporal logic $\tlxy$ which closes the rankers of \cite{WI07} under boolean operations,
also characterizes $UL$. In a subsequent paper \cite{DKL10}, they gave an important property of rankers called \emph{ranker directionality}. We revisit this logic of rankers, giving a mild generalization of the same and study some key properties of rankers such as \emph{convexity}. We shall give direct reductions between \tlxy\/ formulas and \potdfa\/ in both directions and analyse the complexity of translations. This also gives us an \np-complete satisfiability algorithm for \tlxy\/ formulas. 

\subsection{\tlxy: Syntax and Semantics}
\tlxy\/ is a unary deterministic temporal logic with the deterministic modalities $X_a$ (\textit{next}-$a$) and $Y_a$ (\textit{previous}-$a$) which uniquely mark the first and last occurrences (respectively) of a letter $a$ from the given position. We also include their corresponding \textit{weak} modalities ($\weakx{a}$ and $\weaky{a}$), and \emph{unit} modalities ($\xunit,\yunit$) which access the next and previous positions respectively. $SP$ (\textit{Starting Position}) and $EP$ (\textit{Ending Position}) are additional modalities which uniquely determine the first and last positions of the word respectively.

Let $\phi,\phi_1$ and $\phi_2$ range over \tlxy\/ formulas and $a$ range over letters from a finite alphabet $\Sigma$. The syntax of $\tlxy$ is given by: \\
\[
\phi:= a ~\mid~ \top ~\mid~ SP\phi_1 ~\mid~ EP\phi_1 ~\mid~  X_a \phi_1 ~\mid~ Y_a \phi_1 ~\mid~  \weakx{a} \phi_1 ~\mid~ \weaky{a} \phi_1
~\mid~\xunit \phi_1 ~\mid~ \yunit \phi_1
 ~\mid~\phi_1\lor\phi_2 ~\mid~ \neg\phi_1\\
\]
\medskip

$\gabar{a} = \neg X_a\top$ and $\habar{a} = \neg Y_a\top$ are derived atomic formulas. \\

\begin{remark}
The weak modalities and unit modalities do not add expressive power to the logic. They may be derived using the $X_a$ and $Y_a$ modalities alone. However, we include them in the syntax of the logic. As we shall see later in the paper, properties of these generalized rankers play a crucial role in our formulations of reductions between logics for $UL$.  
\end{remark}

A \tlxy\/ formula $\phi$ may be represented by its parse tree $T_{\phi}$ with each node representing a modal or boolean operator such that the subformulas of $\phi$ form the subtrees of $T_{\phi}$. Let $\sub(n)$ denote the subformula corresponding to the subtree rooted at node $n$, and $n$ be labelled by $\opr(n)$ which is the outermost operator (such as $X_a$ or $\lor$) if $n$ is an interior node, and by a letter or $\top$, if it is a leaf node. We will use the notion of subformulas and nodes interchangeably. The \textit{ancestry} of a subformula $n$ is the set of nodes in the path from the root up to (and including) $n$. The depth of a node is its distance from the root.

Semantics of \tlxy\/ formulas is as given below. Let $w \in \Sigma^+$
be a non-empty finite word and let $i \in dom(w)$ be a position within
the word.\\
\begin{tabular}{rcl}
$w,i\models a$ & iff & $w(i)=a$\\
$w,i\models SP \phi$ & iff & $w,1\models \phi$\\
$w,i\models EP\phi$ & iff & $w,\#w\models \phi$\\
$w,i\models X_a\phi$ & iff & $\exists j>i ~.~ w(j)=a$ and $\forall i<k<j. w(k)\neq a$ and $w,j\models\phi$.\\
$w,i\models Y_a\phi$ & iff & $\exists j<i ~.~ w(j)=a$ and $\forall j<k<i. w(k)\neq a$ and $w,j\models\phi$.\\
$w,i\models \weakx{a}\phi$ & iff & $\exists j\geq i ~.~ w(j)=a$ and $\forall i\leq k<j. w(k)\neq a$ and $w,j\models\phi$.\\
$w,i\models \weaky{a}\phi$ & iff & $\exists j\leq i ~.~ w(j)=a$ and $\forall j<k\leq i. w(k)\neq a$ and $w,j\models\phi$.\\
$w,i\models \xunit\phi_1$ & iff & $\exists j=i+1 ~.~ w,j\models \phi_1$\\
$w,i\models \yunit\phi_1$ & iff & $\exists j=i-1 ~.~ w,j\models \phi_1$\\
$w,i\models \phi_1\lor\phi_2$ & iff & $w,i\models\phi_1$ or $w,i\models \phi_2$\\
$w,i\models \neg\phi_1$ & iff & $w,i\not\models\phi_1$\\
\end{tabular} \\

The language accepted by a \tlxy\/ formula $\phi$ is given by $\mathcal L(\phi)=\{w ~\mid~ w,1\models\phi\}$.\\

\subsection{\tlxy: Unique Parsing} 
\tlxy\/ is a \emph{Deterministic Logic}: Given any word $w\in\Sigma^+$ and \tlxy\/ formula  $\phi$, for any
subformula $\eta$ of $\phi$, there exists a unique position in $dom(w)$ where $\eta$ must be evaluated in order to find the truth of $\phi$. This position is denoted by $\pos(\eta)$ and is uniquely determined by the ancestry of $\eta$. 
This property of the logic is referred to as the \textit{unique parsing} property \cite{LPS08}. If such a position does not exist, then $\pos(\eta)=\bot$. It can be defined by induction 
on the depth of $\eta$ as follows. If $\eta_{root}$ is the topmost node denoting the full formula, then $\pos(\eta_{root})= 1$. 
Inductively, if $\eta =op(\eta_1)$ or $\eta=op(\eta_1,\eta_2)$ and $\pos(\eta)=\bot$ then
$\pos(\eta_1)=\pos(\eta_2)=\bot$. For the remaining cases, let $\pos(\eta)=i$ (which is not $\bot$). Then,
\begin{itemize}
\item If $\eta = SP \eta_1$, then $\pos(\eta_1) = 1$.
 \item If $\eta= EP \eta_1$ then $\pos(\eta_1) = \#w$.
\item If $\eta= X_a \eta_1$. Then, $\pos(\eta_1) = \bot$ if $\forall k > i, ~w(k)\neq a$. \\
Otherwise, $\pos(\eta_1) = j$ s.t.
$j>i$ and $w(j)=a$ and $\forall i < k < j, ~w(k)\neq a$.
\item If $\eta= Y_a \eta_1$. Then, $\pos(\eta_1) = \bot$ if $\forall k < i, ~w(k)\neq a$. \\
Otherwise, $\pos(\eta_1) = j$ s.t.
$j<i$ and $w(j)=a$ and $\forall j < k < i, ~w(k)\neq a$.
\item If $\eta= \weakx{a} \eta_1$. Then, $\pos(\eta_1) = \bot$ if $\forall k \geq i, ~w(k)\neq a$. \\
Otherwise, $\pos(\eta_1) = j$ s.t.
$j\geq i$ and $w(j)=a$ and $\forall i \leq k < j, ~w(k)\neq a$.
\item If $\eta= \weaky{a} \eta_1$. Then, $\pos(\eta_1) = \bot$ if $\forall k \leq i, ~w(k)\neq a$. \\
Otherwise, $\pos(\eta_1) = j$ s.t.
$j\leq i$ and $w(j)=a$ and $\forall j < k \leq i, ~w(k)\neq a$.
\item If $\eta= \xunit\eta_1$. Then $\pos(\eta_1)=\bot$ if $i=\#w$\\
Otherwise, $\pos(\eta_1)~=~ i+1$
\item If $\eta= \yunit\eta_1$. Then $\pos(\eta_1)=\bot$ if $i=1$\\
Otherwise, $\pos(\eta_1)~=~ i-1$
\item If $\eta= \eta_1 \lor \eta_2$ or $\eta = \eta_1 \land \eta_2$ then $\pos(\eta_1)=\pos(\eta_2)=\pos(\eta)$. Similarly,
if $\eta= \neg \eta_1$ then $\pos(\eta_1) = \pos(\eta)$. 
\end{itemize}

\begin{example}
\label{exam:tlxy}
Consider the language given by $R=\Sigma^* a c^* d \{b,c,d\}^*$ as in Example \ref{exam:potdfa} of Chapter \ref{chap:ulpotdfa}. The language defines the set of all words such that the last $a$ in the word has a successive $d$ such that there is no $b$ between them. This may equivalently be expressed using the \tlxy\/ formula 
\[
 \phi ~:=~ EP ~\weaky{a} X_d (\neg Y_b\top ~~\lor~~ Y_bX_a\top)
\]
For any word $w$ which belongs to the language of the above formula, $\pos(X_d \neg(Y_bX_a\top))$ matches with the last $a$ in the word. Let this position be $i$. Further, $\pos(Y_bX_a\top)$ is a position $j$ such that $j$ is the first $d$ after $i$. Now at $j$, the formula $(\neg Y_b\top ~~\lor~~ Y_bX_a\top)$ holds if and only if either there is no $b$ before $j$ or the $b$ before $j$ (which is at some $k$), is such that there is an $a$ after it. Hence $k<i$, and there is no $b$ between $i$ and $j$. Hence we can see that the above formula $\phi$ expresses the language given by $R$.
\end{example}

\subsection{Ranker Formulas}
The notion of {\em rankers} \cite{WI07} has played an important role in characterizing unambiguous languages $UL$. They were originally introduced as turtle programs by Schwentick {\em et al} \cite{STV01}. Basically a ranker $r$ is a finite sequence of instructions of the form $X_a$ (denoting ``go to the next $a$ in the word'') or $Y_a$ (denoting ``go to the previous $a$ in the word'').
Given a word $w$ and a starting position $i$, the execution of a ranker $r$ succeeds and ends
at a final position $j$ if all the instructions find their required letter. This is denoted by
$w,i \models r$. 

Here, we generalize rankers and call them \stls\/. These are essentially $\tlxy$ formulas without any boolean operators, but including both the strict and the non-strict deterministic modalities ($X_a,Y_a,\weakx{a},\weaky{a}$), the unit-step modalities ($\xunit,yunit$), as well as the end postion modalities ($SP,EP$). This generalization maintains the key deterministic nature of rankers. \\
\\
The syntax of \stls\/ is as follows: 
\\
\tab $\phi := \top ~\mid~ SP\phi ~\mid~ EP\phi ~\mid~ X_a\phi ~\mid~ Y_a\phi ~\mid~ \weakx{a}\phi ~\mid~ \weaky{a}\phi ~\mid~ \xunit\phi ~\mid~\yunit\phi$
\footnote{While $a$ (for every $a\in\Sigma$) is an atomic formula in the case of \tlxy\/ formulas, \stls\/ do not have $a$ as an atomic formula.}
\\
Given a \stl\/ $\psi$, let $\leafn(\psi)$ denote the unique leaf node in $T_{\psi}$. Note that the parse tree of \stls\/ comprise of a single path, giving unique $\leafn(\psi)$ and $\opr(\leafn(\psi))=\top$.
For a given word $w$, the position of leaf node is denoted as $\lpos(\psi) = \pos(\leafn(\psi))$.

\paragraph{Ranker Directionality\\}
Consider a \stl\/ $\psi$. We can construct  \tlxy\/ formulas $\pfless(\psi)$, $\pfleq(\psi)$, $\pfgreat(\psi)$, $\pfgeq(\psi)$ such that they satisfy the following Lemma \ref{lem:pfpos}. These formulas are called \emph{ranker directionality formulas} and they allow us to analyse the relative positioning of the current position, with respect to the $lpos$ of the ranker. These formulas were given by \cite{DKL10} for rankers. We generalize them for \stls. 

Let $\phi\top$ be a \stl\/ where $\phi$ is the ancestor of the leaf node $\top$. The ranker directionality formulas are given by Table \ref{tab:trank}, by induction on the length of the ranker. In this table, let $\mathit{Atfirst}\defn ~\neg(\lor_{a\in\Sigma}(Y_a\top))$ and $\mathit{Atlast}\defn ~\neg(\lor_{a\in\Sigma}(X_a\top))$ be formulas which hold exactly at the first and last positions in any word. Since every \stl\/ formula is evaluated starting from the beginning of the word, we shall assume that at the top level the ranker begins with the $SP$ modality. 
\begin{table}
\begin{tabular}{|c|c|c|c|c|}
\hline
$\psi$ & $\pfless(\psi)$ & $\pfleq(\psi)$ & $\pfgreat(\psi)$ & $\pfgeq(\psi)$\\
\hline
\hline
$\phi SP\top$& $\bot$ & $\mathit{Atfirst}$ & $\neg\mathit{Atfirst}$ & $\top$\\
\hline
$\phi EP\top$& $\neg\mathit{Atlast}$ & $\top$ & $\bot$ & $\mathit{Atlast}$\\
\hline
$\phi\weakx{a}\top$ & $X_a(\pfleq(\psi))$ & $\habar{a}\lor (Y_a\pfless(\phi\top))$ & $Y_a\pfgeq(\phi\top)$ & $\gabar{a}\lor X_a\pfgreat(\psi)$\\
\hline
$\phi X_{a}\top$ & $X_a(\pfleq(\psi))$ & $\habar{a}\lor (Y_a\pfleq(\phi\top))$ & $Y_a\pfgreat(\phi\top)$ & $\gabar{a}\lor X_a\pfgreat(\psi)$\\
\hline
$\phi\weaky{a}\top$ & $X_a\pfleq(\phi\top)$ & $\habar{a}\lor (Y_a\pfless(\psi))$ & $Y_a\pfgeq(\psi)$ & $\gabar{a}\lor X_a\pfgreat(\phi\top)$\\
\hline
$\phi Y_a\top$ & $X_a\pfless(\phi\top)$ & $\habar{a}\lor (Y_a\pfless(\psi))$ & $Y_a\pfgeq(\psi)$ & $\gabar{a}\lor X_a\pfgeq(\phi\top)$\\
\hline
$\phi\xunit\top$ & $\pfleq(\phi\top)$ & $\mathit{Atfirst} ~\lor~ \yunit\pfleq(\phi\top)$ & $\yunit \pfgreat(\phi\top)$ & $\pfgreat(\phi\top)$ \\
\hline
$\phi\yunit\top$ & $\xunit\pfless(\phi\top)$ & $\pfless(\phi\top)$ & $\pfgeq(\phi\top)$ & $\mathit{Atlast} ~\lor~\xunit\pfgeq(\phi\top)$ \\
\hline 
\end{tabular}
\caption{Ranker Directionality Formulas}
\label{tab:trank}
\end{table}

Observe that the size of the ranker directionality formula is linear in the size of the \stl.\\

\begin{lemma}[Ranker Directionality\cite{DKL10}] 
\label{lem:pfpos}
Let $\psi$ be a \stl\/. Then $\forall w\in\Sigma^+$ and $\forall i\in dom(w)$, if $\lpos(\psi)\neq\bot$, then
\begin{itemize}
\item $w,i\models\pfless(\psi)$ iff $i<\lpos(\psi)$
\item $w,i\models\pfleq(\psi)$ iff $i\leq \lpos(\psi)$
\item $w,i\models\pfgreat(\psi)$ iff $i>\lpos(\psi)$
\item $w,i\models\pfgeq(\psi)$ iff $i\geq \lpos(\psi)$
\end{itemize}
\end{lemma}
\begin{proof}
The correctness of the construction of the ranker directionality formulas is a direct consequence of the semantics of \tlxy. We shall prove some key cases from Table \ref{tab:trank}. Consider any $w\in\Sigma^+$ and for all the cases below, assume $\lpos(\psi)\neq\bot$.
\begin{itemize}
 \item Consider $\psi=\phi\weakx{a}\top$. This is depicted in Figure \ref{fig:weakx}. Note that there are two mutually exclusive cases: (i) If $w(\lpos(\phi\top)=a$ then $\lpos(\phi\top) = \lpos(\psi)$. (ii) If $w(\lpos(\phi\top)\neq a$ then $\lpos(\psi)~ >~ \lpos(\phi\top)$.\\
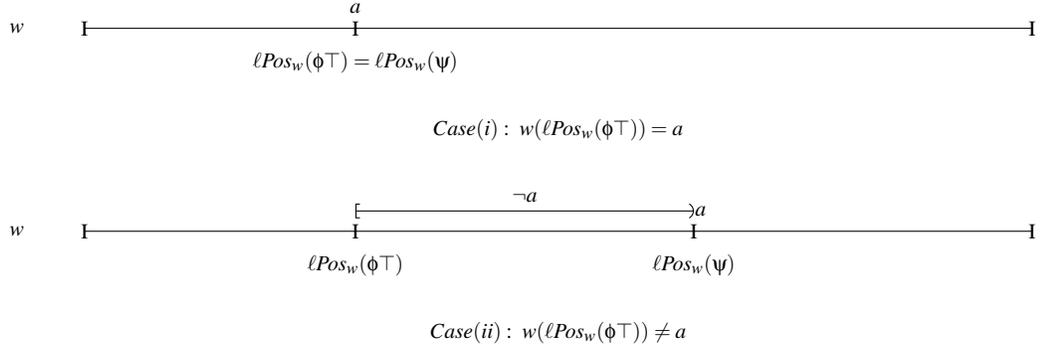
\begin{figure}
 \begin{tikzpicture}[scale=0.9,transform shape]
 \draw (8,3.5) node {$Case (i): ~ w(\lpos(\phi\top))= a$};
\draw (0,5) node {$w$}; \draw (1,5) node{I}-- (5,5) node{I}--  (15,5) node{I};
\draw (5,4.5) node{$\lpos(\phi\top)=\lpos(\psi)$};
 \draw (5,5.3) node {$a$};
 \draw (8,0.5) node {$Case (ii): ~ w(\lpos(\phi\top))\neq a$};
\draw (0,2) node {$w$}; \draw (1,2) node{I}-- (5,2) node{I}-- (10,2) node{I}-- (15,2) node{I};
\draw (5,1.5) node{$\lpos(\phi\top)$}; \draw (10,1.5) node {$\lpos(\psi)$};
\draw [[-)] (5,2.3)--(10,2.3); \draw (7.5,2.5) node {$\neg a$}; \draw (10.1,2.3) node {$a$};
 \end{tikzpicture}
\caption{$\psi=\phi\weakx{a}\top$}
\label{fig:weakx}
\end{figure}

\begin{tabular}{lll}
$\forall i ~.~ i\leq \lpos(\psi)$ & iff & either there exists no $a$ to the left of $i$\\
&& otherwise, the last $a$ strictly to the left of $i$, is strictly \\
&& to the left of $\lpos(\phi\top)$\\
&iff& $\habar{a}\lor (Y_a\pfless(\phi\top))$\\
\end{tabular}

\item Consider $\psi= \phi X_{a}\top$. This is depicted in Figure \ref{fig:strictx}. Note that $\lpos(\phi\top) ~<~ \lpos(\psi)$.\\
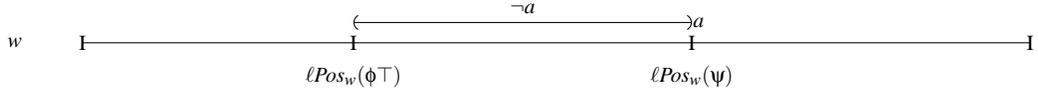
\begin{figure}
 \begin{tikzpicture}[scale=0.9,transform shape]
\draw (0,2) node {$w$}; \draw (1,2) node{I}-- (5,2) node{I}-- (10,2) node{I}-- (15,2) node{I};
\draw (5,1.5) node{$\lpos(\phi\top)$}; \draw (10,1.5) node {$\lpos(\psi)$};
\draw [(-)] (5,2.3)--(10,2.3); \draw (7.5,2.5) node {$\neg a$}; \draw (10.1,2.3) node {$a$};
 \end{tikzpicture}
\caption{$\psi=\phi X_a\top$}
\label{fig:strictx}
\end{figure}
\begin{tabular}{lll}
$\forall i ~.~ i\leq\lpos(\psi)$ & iff & either there exists no $a$ to the left of $i$\\
&& otherwise, the last $a$ strictly to the left of $i$ is $\leq$ $\lpos(\phi\top)$\\
&iff& $\habar{a}\lor (Y_a\pfleq(\phi\top))$\\
\end{tabular}

\item Consider $\psi=\phi\xunit\top$. This is depicted in Figure \ref{fig:xunit}. Note that $\lpos(\psi) = \lpos(\phi\top) +1$.\\
\begin{figure}
\begin{tikzpicture}[scale=0.7,transform shape]
\draw (0,2) node {$w$}; \draw (1,2) node{I}-- (5,2) node{I}-- (6,2) node{I}-- (15,2) node{I};
\draw (4,0.5) node{$\lpos(\phi\top)$}; \draw (7,0.5) node {$\lpos(\psi)$};
\draw [->](4,0.7)--(5,1.8); \draw [->](7,0.7)--(6,1.8);
\draw (5,2.3) node {$x$}; \draw (6,2.3) node {$(x+1)$};
\end{tikzpicture}
\caption{$\psi=\phi\xunit\top$}
\label{fig:xunit}
\end{figure}
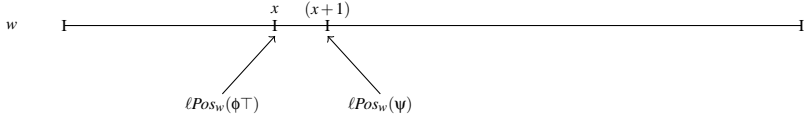 
\begin{tabular}{lll}
$\forall i~.~ i\leq\lpos(\psi)$ & iff & $(i-1)\leq \lpos(\phi\top)$\\
&& iff either $i=1$(since $\lpos(\psi)> 1$) or $(i-1)\leq \lpos(\phi)$\\
&& iff $w,i\models \mathit{Atfirst} ~\lor~ \yunit\pfleq(\phi)$\\
\end{tabular}
\end{itemize}
\end{proof}

\paragraph{From \tlxy\/ to \stls\/\\}
We shall show that every \tlxy\/ formula may be written as a boolean combination of \stls\/ and atomic formulas. This is done by first eliminating atomic formulas of the form $a$ for any $a\in\Sigma$ and then ``pulling out'' booleans. This is given in the proposition below.\\

\begin{proposition}
\label{prop:simbool}
For any \tlxy\/ formula $\phi$, there is a boolean combination $\mathcal B(\psi_i)$ of formulas $\psi_i$, such that $\mathcal L(\phi) = \mathcal L(\mathcal B(\psi_i))$. Each $\psi_i$ is either an atomic formula or \stl. Moreover each $\psi_i$ is linear in the size of $\phi$.
\end{proposition}
\begin{proof}
Every boolean may be ``pulled out'' of \tlxy\/ formulas using the equivalences below.
\begin{itemize}
\item $\phi X_a (\phi_1 \lor \phi_2) \equiv (\phi X_a\phi_1) \lor (\phi X_a\phi_2)$
\item $\phi X_a (\phi_1 \land \phi_2) \equiv (\phi X_a\phi_1) \land (\phi X_a\phi_2)$
\item $\phi X_a(\neg \phi_1) \equiv  \neg(\phi X_a\phi_1)\land \phi X_a\top$
\item $\phi Y_a (\phi_1 \lor \phi_2) \equiv (\phi Y_a\phi_1) \lor (\phi Y_a\phi_2)$
\item $\phi Y_a (\phi_1 \land \phi_2) \equiv (\phi Y_a\phi_1) \land (\phi Y_a\phi_2)$
\item $\phi Y_a(\neg \phi_1) \equiv  \neg(\phi Y_a\phi_1)\land \phi Y_a\top$ \\
\end{itemize}

Now, if $\psi$ is a \stl, define formulas\\
$\mathit{next}(\psi) ~=~ \neg\bigvee\limits_{b\in\Sigma}(Y_b\land \pfgreat(\psi))$ \\
$\mathit{prev}(\psi) ~=~ \neg\bigvee\limits_{b\in\Sigma}(X_b\land \pfless(\psi))$ \\
Observe that $\forall w\in \Sigma^*$ such that $\lpos(\psi)\neq\bot$, 
\begin{itemize}
\item If $i>\lpos(\psi)$, then $w,i\models \mathit{next}(\psi)$ if and only if $i=\lpos(\psi)+1$
\item If $i<\lpos(\psi)$, then $w,i\models \mathit{prev}(\psi)$ if and only if $i=\lpos(\psi)-1$
\end{itemize}
In other words, given a \stl\/ $\psi$, the formulas $\mathit{next}(\psi)$ and $\mathit{prev}(\psi)$ respectively hold exactly at the position next to and previous to $\lpos(\psi)$. 

The atomic formula $a$ may be eliminated from the \stls\/ using the equivalences:
\begin{itemize}
\item $\phi X_ba ~\equiv~ \phi X_b\top$ and $\phi Y_ba~\equiv~ \phi Y_b\top$ if $a=b$
\item $\phi X_ba ~\equiv~ \bot$ and $\phi Y_ba \equiv \bot$ if $a\neq b$
\item $\phi \weakx{b}a ~\equiv~ \phi \weakx{b}\top$ and $\phi\weaky{b}a~\equiv~ \phi\weaky{b}\top$ if $a=b$
\item $\phi \weakx{b}a ~\equiv~ \bot$ and $\phi\weaky{b}a ~\equiv~ \bot$ if $a\neq b$
\item $\phi SP a~\equiv~ \phi SP\weakx{a}(Atfirst)$
\item $\phi EP a~\equiv~ \phi EP\weaky{a}(Atlast)$
\item $\phi \xunit a ~\equiv~ \phi X_a\mathit{next}(\phi)$
\item $\phi \yunit a ~\equiv~ \phi Y_a \mathit{prev}(\phi)$\\
\end{itemize}
After elimination of atomic formulas, we obtain \tlxy\/ formulas with booleans. We may again eliminate booleans using the equivalencies given above. The resulting formula is a boolean function $\mathcal B(\psi)$ where each $\psi$ is either an atomic formula or a \stl\/ of size linear in $\phi$.
\end{proof}

\begin{example}
\label{exam:bool}
We may eliminate the negation and conjunctions from the formula as given below:\\
\begin{tabular}{ll}
$ \phi ~:=~ EP \weaky{a} X_d [\neg(Y_bX_a\top) ~\land~ Y_c\top]$ & $\equiv~ EP \weaky{a} X_d [\neg(Y_bX_a\top)] ~\land ~ EP \weaky{a} X_dY_c\top$\\
 & $\equiv~ [\neg( EP \weaky{a} X_d Y_bX_a\top)~\land  EP \weaky{a} X_d \top] ~\land~ EP \weaky{a} X_dY_c\top$\\
\end{tabular}
\end{example}

\paragraph{Eliminating additional modalities}
\begin{proposition}
\label{prop:genrank}
Every \tlxy\/ formula may be expressed as language-equivalent \tlxy\/ formula without weak modalities and unit-step modalities.
\end{proposition}
\begin{proof}
Consider any \tlxy\/ formula $\Phi$. We shall reduce it to a formula without weak modalities and unit-step modalities. 
Firstly, we may pull out the booleans to reduce the formula to a boolean combination of \stls (using Proposition \ref{prop:simbool}). We may then eliminate the unit-step modalities from the \stls\/ using the following rules:
\begin{itemize}
\item[] $\phi_1\xunit\phi_2 ~\equiv~ \phi_1\bigvee\limits_{a\in\Sigma}[X_a(\mathit{next}(\phi_1)\land \phi_2)]$
\item[] $\phi_1\yunit\phi_2 ~\equiv~ \phi_1\bigvee\limits_{a\in\Sigma}[Y_a(\mathit{prev}(\phi_1)\land \phi_2)]$
\end{itemize}
Note that eliminating each unit step modality in a \tlxy\/ formula involves first pulling out booleans and then applying one of the above rules to each \stl\/. This is because the $\mathit{next}$ and $\mathit{prev}$ formulas use ranker directionality formulas which are applicable to \stls\/ and not \tlxy\/ formulas in general. 

Further, we may eliminate the weak modalities using the following reductions:
\begin{itemize}
\item[] $\weakx{a}\phi ~\equiv ~ (a\land \phi) \lor (\neg a\land X_a\phi)$
\item[] $\weaky{a}\phi ~\equiv ~ (a\land \phi) \lor (\neg a\land Y_a\phi)$
\end{itemize}
\end{proof}

\paragraph{Convexity of \stls\/\\}
We show here another useful property of \stls\/, which will be important in reductions given later in the paper.
\begin{lemma}[Convexity]
\label{lem:convex1}
For any \stl\/ $\psi$, and any word $w\in\Sigma^+$, if there exist $i,j\in dom(w)$ such that $i<j$ and $w,i\models\psi$ and $w,j\models\psi$, then $\forall i<k<j$, we have $w,k\models\psi$.
\end{lemma}
\begin{proof}
We prove the lemma by induction on the structure of $\psi$. The lemma trivially holds for the base case of $\psi=\top$. 
We give the inductive argument for the case of $\psi= X_a\phi$ (other cases are similar/simpler and omitted). Assume that the lemma holds true for $\phi$ (Induction Hypothesis). Let $i,j\in dom(w)$ such that $i<j$ and $w,i\models\psi$ and $w,j\models \psi$. Consider some $k$ such that $i<k<j$. Let $i'$ and $j'$ respectively  be the positions of first occurrence of $a$ after $i$ and $j$.
These positions must exist as $w,i\models\psi$ and $w,j\models \psi$ and we have $w,i'\models\phi$ and $w,j'\models\phi$ and
$i<i' \leq j'$ with $j < j'$. Hence, $j'>k$. Let $k'$ be the position of first occurrence of $a$ after $k$. Such a position must exist since $w(j')=a$ and $j'>k$. Also $i' \leq k' \leq j'$. Then by induction hypothesis, $w,k' \models \phi$ and hence $w,k \models \psi$.
\end{proof}

\paragraph{Sequential composition of Rankers\\}
Through the rest of this chapter, we shall alternatively use the terms ``ranker'' and ``\stl''. We say that a ranker $\phi$ \textit{accepts} at a position $i$ in a word $w$ if $\lpos(\phi)=i$. Given a ranker $\phi_1$ and any \tlxy\/ formula $\phi_2$, denote by $\phi_1;\phi_2$ the \tlxy\/ formula obtained by replacing the leaf node of $\phi_1$ by the parse tree of $\phi_2$. Hence, it is easy to see that for any word $w$, $w,1\models \phi_1;\phi_2$ iff $w,i\models\phi_2$, where $i=\lpos(\phi_1)$. Note that if $\phi_1$ and $\phi_2$ are \stls\/ then $\phi_1;\phi_2$ is also a \stl.

\subsection{Equivalence of \tlxy\/ and \potdfa\/}
We give a language-preserving reductions from \tlxy\/ to \potdfa\/ and analyse its complexity. This also gives us an \np-complete language non-emptiness checking algorithm for \tlxy\/ formulas.

\paragraph{From \tlxy\/ to \potdfa\\}
First, we shall show a language-preserving conversion from \tlxy\/ formulas to \potdfa. 
One simple approach is to convert each ranker without weak or unit modalities into \potdfa. Since every $\phi$ can be written as a boolean combination of such \stls\/ and since \potdfa\/ are effectively closed under boolean operations, we obtain a language-equivalent automaton. However, the resulting automaton is exponential in size of $\phi$. 
Below, we obtain a polynomial-sized automaton by utilizing the unique parsability property of \tlxy\/ formulas.

\begin{theorem}
\label{thm:tlxyaut}
Given any \tlxy\/ formula $\phi$ we may construct an equivalent \potdfa\/ $\autm(\phi)$ such that $\mathcal L(\phi) = \mathcal L(\autm(\phi))$. The number of states in $\autm(\phi)$ is polynomial in the size $\phi$. 
\end{theorem}

\paragraph{Construction\\}
The efficient reduction from \tlxy\/ to \potdfa\/ relies on the property of unique parsing of \tlxy formulas. We use the \ete\/ representation to illustrate the construction of the \potdfa. Fix a \tlxy\/ formula $\Phi$. For any subformula $\phi$ of $\Phi$ and any given word $w$, $\pos(\phi)$ depends on the context of $\phi$ and may be evaluated in a top-down manner. We construct an \ete\/ $POS(\phi)$ which is given by the following proposition.

\begin{proposition}
For any subformula $\phi$ of $\Phi$ and any word $w\in\Sigma^*$, we have 
\begin{itemize}
\item $POS(\phi)(w,1) = (t,i)$ iff $\pos(\phi)=i$
\item $POS(\phi)(w,1) = (f,i)$ iff $\pos(\phi)=\bot$
\end{itemize}
\end{proposition}
\begin{proof}
The \ete\/ for $POS(\phi)$ may be constructed by structural induction on the formula as follows.
\begin{itemize}
\item $POS(\Phi) ~=~ \bscan{\lend}{\Sigma'} ; (\funit{\lend})$
\item If $\phi = X_a\phi_1$ then $POS(\phi_1) ~=~ POS(\phi);\funit{\Sigma'};\fscan{a}{\Sigma'}$
\item If $\phi = Y_a\phi_1$ then $POS(\phi_1) ~=~ POS(\phi);\bunit{\Sigma'};\bscan{a}{\Sigma'}$
\item If $\phi = \weakx{a}\phi_1$ then $POS(\phi_1) ~=~ POS(\phi);\fscan{a}{\Sigma'}$
\item If $\phi = \weaky{a}\phi_1$ then $POS(\phi_1) ~=~ POS(\phi);\bscan{a}{\Sigma'}$
\item If $\phi = \xunit\phi_1$ then $POS(\phi_1) ~=~ POS(\phi) ~;~ [(\funit{\Sigma};\bunit{\rend}) ~?~ Rej ~:~ \funit{\Sigma}]$
\item If $\phi = \yunit\phi_1$ then $POS(\phi_1) ~=~ POS(\phi) ~;~ [(\bunit{\Sigma};\funit{\lend}) ~?~ Rej ~:~ \bunit{\Sigma}]$
\item If $\phi = SP\phi_1$ then $POS(\phi_1) ~=~ \bscan{\lend}{\Sigma'} ~;~(\funit{\lend})$
\item If $\phi = EP\phi_1$ then $POS(\phi_1) ~=~ \fscan{\rend}{\Sigma'} ~;~ (\bunit{\rend})$
\item If $\phi= \phi_1\lor\phi_2$ then $POS(\phi_1)~=~POS(\phi_2)=POS(\phi)$
\item If $\phi = \neg \phi_1$ then $POS(\phi_1) ~=~ POS(\phi) $
\end{itemize}
The correctness of the above construction may be directly deduced from the definition of $\pos(\phi)$ for \tlxy\/ formulas. Note that the \ete\/ for $POS(\phi_1)$ when $\phi=\xunit\phi_1$ is constructed as follows. It first checks if $POS(\phi)$ is at the last position in the word (by using $\funit{\Sigma};\bunit{\rend}$). If so, it rejects (evaluates to $f$), in which case $\pos(\phi_1)=\bot$. Otherwise, it accepts at the next position after $POS(\phi)$. The case of $\phi=\yunit\phi_1$ is symmetric to this.
By observing the above construction, the following property may be easily verified.
\end{proof}

Now, for every subformula $\phi$, we construct \ete\/ $EVAL(\phi)$ which evaluates the formula at is unique position, as follows.
\begin{proposition}
For any subformula $\phi$ of $\Phi$ and any word $w\in\Sigma^*$ we have $EVAL(w,1) = (t,i)$ iff $\pos(\phi)\neq\bot$ and $w,\pos(\phi)\models\phi$.
\end{proposition}
\begin{proof}
\begin{itemize}
\item If $\phi=\top$ then $EVAL(\phi) ~=~ POS(\phi);Acc$
\item If $\phi = X_a\phi_1, Y_a\phi_1,\weakx{a}\phi_1,\weaky{a}\phi_1, SP\phi_1,EP\phi_1,\xunit\phi_1$ or $\yunit\phi_1$ then \\
\tab $EVAL(\phi) ~=~ POS(\phi_1);EVAL(\phi_1)$
\item If $\phi= \phi_1\lor\phi_2$ then $[POS(\phi);EVAL(\phi_1)]~?~ [Acc] ~:~ [POS(\phi); EVAL(\phi_2)]$
\item If $\phi = \neg \phi_1$ then $EVAL(\phi_1) ~?~ Rej ~:~ Acc$
\end{itemize} 
Hence, we may verify that for any subformula $\phi$ and any word $w$, $EVAL(w,1) = (t,i)$ iff $\pos(\phi)\neq\bot$ and $w,\pos(\phi)\models\phi$. 
\end{proof}

For the top level formula, we can see that $EVAL(\Phi)$ is the language-equivalent \ete\/ for $\Phi$.
\paragraph{Complexity\\}
Consider a \tlxy\/ formula $\Phi$ of length $l$. For every subformula $\phi$ of $\Phi$, observe that $POS(\phi)$ is linear in $l$. Further, $EVAL(\phi)$ is polynomial in $l$. Therefore, we can conclude that the size of the \ete (and hence the \potdfa) which is language-equivalent to $\Phi$ is polynomial in the size of $\Phi$. Hence the theorem (Theorem \ref{thm:tlxyaut}). 

The above translation allows us to give a tight \np-complete satisfiability complexity for \tlxy\/ formulas. We may convert a given \tlxy\/ formula to its language-equivalent \potdfa\/ whose size is polynomial in the size of its original formula. Since language emptiness of a \potdfa\/ is an \np-complete problem, satisfiability problem of \tlxy\/ is in \np. The \np-hardness of the satisfiaility problem of \tlxy\/ can be inferred from the \np-complete satisfiability of propositional temporal logic. Hence the following theorem.

\begin{theorem}[Satisfiability of \tlxy\/ formulas]
The satisifability of \tlxy\/ formulas is decidable with \np-complete complexity.
\end{theorem}

\section{\duds}
The deterministic Until-Since logic \duds\/ in some sense is very close to the \potdfa\/ automata: the looping of the automaton in a state until a progress transition is enabled, corresponds well with the invariance and eventuality conditions of the until and since modalities. 

Let $A\subseteq\Sigma$, $a,b\in\Sigma$ and $\phi$ range over \duds\/ formulas. A \duds formula may be given by the following syntax.
\[
\top ~\mid~ a ~\mid~ A\detu_b\phi ~\mid~ A\dets_b\phi ~\mid~ \phi\lor\phi ~\mid~ \neg\phi
\]
Given a word $w\in \Sigma^*$, and $i\in dom(w)$, \duds\/ formulas may be interpreted using the following rules. 
\begin{center}
\begin{tabular}{rcl}
$w,i\models a$ & iff &$w(i)=a$\\
$w,i\models A\detu_b\phi$ & iff & $\exists j>i~.~ w(j)=b\land \forall i<k<j ~.~ w(k)\in A\setminus b ~\land~ w,j\models\phi$\\
$w,i\models A\dets_b\phi$ & iff & $\exists j<i~.~ w(j)=b\land \forall j<k<i ~.~ w(k)\in A\setminus b ~\land~ w,j\models\phi$\\
\end{tabular}
\end{center}
The boolean operators have their usual meaning. The language defined by a \duds\/ formula $\phi$ is given by $\mathcal L(\phi)= \{w\in\Sigma^* ~\mid~ w,1\models\phi\}$ (if the outermost operator of $\phi$ is a $\detu$ operator) and $\mathcal L(\phi)= \{w\in\Sigma^* ~\mid~ w,\#w\models\phi\}$ (if the outermost operator of $\phi$ is a $\dets$ operator). \duds\/ formulas may be represented as a DAG, in the usual way, with the modal/boolean operators at the intermediate nodes. 

\begin{example}
The language described in Example \ref{exam:potdfa} which is given by $\Sigma^* a c^* d \{b,c,d\}^*$ may be expressed using the \duds\/ formula $\Sigma\dets_a~(\Sigma\setminus\{b\}~\detu_d\top)$. 
\end{example}

\paragraph{\duds\/ and Unique Parsability}
The $\detu$ and $\dets$ modalities of \duds\/ are deterministic, in the sense that they uniquely define the position at which its subformula must be evaluated. Hence, for every subformula $\psi$ of a \duds\/ formula $\phi$, and any word $w$, there exists a unique position denoted as $\pos(\psi)$, where $\psi$ is to be evaluated. Moreover, $\pos(\psi)$ is determined by the context of $\psi$ in $\phi$. For example, consider the subformula $\psi = A\detu_b (\psi')$, such that $\pos(\psi)=i$. Then $\pos(\psi') = j$ such that $j>i$, $w(j)=b$ and $\forall i<k<j ~.~ w(k)\in A\setminus\{b\}$. 

The until and since modalities of \duds\/ seem to subsume the $X_a$ and $Y_a$ modalities of \tlxy: for example $X_a\phi \equiv \Sigma\detu_a\phi$. However both logics share the same expressive power. 

\subsection{From \potdfa\/ to \duds}
The deterministic \emph{until} and \emph{since} operators of \duds\/ naturally model the constraints on the run of  a \potdfa: the looping of the \potdfa\/ in a given state and on a subset of letters until an outward transition is enabled is straightforwardly captured by the invariance condition of the $\detu$ and $\dets$ modalities. We shall now give a translation from \potdfa\/ automata to language-equivalent \duds\/ formulas.

\begin{figure}
\begin{center}
\begin{tikzpicture}[align=center]
\draw (10,5) node (Q) [circle,draw] { $\overrightarrow{q}$ };
\draw (13,6) node (Q1)[circle,draw] {$q_1$};
\draw (13,4) node (Qn)[circle,draw] {$q_n$};
\draw [->] (Q) --node[anchor=south]{$b_1$} (Q1);
\draw [->] (Q) --node[anchor=south]{$b_n$} (Qn);
\draw [->] (8,6) -- (Q);
\draw [->](8,4) -- (Q);
\draw [dotted] (13,5.3)--(13,4.7);
\draw [dotted] (8,5.3)--(8,4.7);
\end{tikzpicture}
\caption{From \potdfa\/ to \duds}
\label{fig:potdfaduds}
\end{center}
\end{figure}

We shall construct a \duds\/ formula $Form(q)$ for each state of $\autm$, such that the following lemma is satisfied.

\begin{lemma}
\label{lem:potdfaduds}
Given a \potdfa\/ $\autm$ and any non-initial state $q$ of $\autm$, we may construct a \duds\/ formula $Form(q)$ such that for every $w\in\Sigma^+$, if $q$ is entered on reading a position $x\in dom(w)$, then $w,x\models Form(q)$ if and only if the run terminates in the accepting state.
\end{lemma}
\begin{proof}
We shall prove this lemma by constructing the formula $Form(q)$ for every non-initial state $q$ in $\autm$. From the syntax of \potdfa\/ it is straightforward to infer that $Form(t)=\top$ and $Form(r)=\bot$. Now, consider a non-initial state $q$ of a \potdfa\/ as shown in Figure \ref{fig:potdfaduds}, such that $q\not\in\{t,r\}$ and $A_q=\Sigma\setminus \{b_1\cdots b_n\}$ is the set of letters on which $q$ loops. Let us assume that $Form(q_1), \cdots Form(q_n)$ are appropriately constructed. If $q\in Q_L$ (i.e. $q$ is a state entered from the left, and the head of the automaton moves right on all transitions whose target state is $q$), then the automaton ``scans'' rightwards from $x$, looping in $q$ on letters from $A_q$, until a progress transition from one of the letters from $\{b_1,\cdots b_n\}$ is enabled. Hence, a progress transition $b_i$ is enabled from $q$ if and only if there exists $y>x$ such that $w(y)=b_i$ and for all $x<k<y$, $w(k)\in A_q$. Further, this run is accepting if and only if $w,y\models Form(q_i)$. 

From the above argument, we may construct $Form(q)$ as follows.
\begin{itemize}
\item If $q\in Q_L$, then
\[
Form(q) ~=~ \bigvee\limits_{i\in\{1,\cdots n\}} [A_q\detu_{b_i}Form(q_i)]
\]
\item If $q\in Q_R$, then
\[
Form(q) ~=~ \bigvee\limits_{i\in\{1,\cdots n\}} [A_q\dets_{b_i}Form(q_i)]
\]
\end{itemize}
\qed
\end{proof}

\begin{theorem}
Given a \potdfa\/  $\autm$, we may construct a \duds\/ formula $Trans(\autm)$ such that $\mathcal L(\autm) ~=~ \mathcal L(Trans(\autm))$, whose DAG representation is linear in the size of $\autm$.
\end{theorem}
\begin{proof}
Consider the start state of the \potdfa\/ $\autm$ which loops on the letters in $A_s$ until a progress transition on one of the letters in $\{c_1, \cdots c_l\}$ is enabled, such that the transition on $c_i$ is targeted into a state $q_i$, for each $i\in\{1\cdots l\}$. From an argument similar to the one in Lemma \ref{lem:potdfaduds}, we may infer that 
\[
Trans(\autm) ~=~ \bigvee\limits_{i\in \{1\cdots l\}} [c_i\land Form(q_i)] ~\lor~ \bigvee\limits_{i\in \{1\cdots l\}} [\bigvee\limits_{b\in A_s} b\land A_s\detu_{c_i}Form(q_i)]
\]
In the above formula, the two sets of disjunctions correspond to the cases when the progress transition from $s$ to the target state is taken on the first position in the word, or any other position, respectively.

In the DAG representation of the formula $Trans(\autm)$ as per the above construction, note that the number of nodes in the DAG is linear in the number of states in $\autm$. This is because $Form(q)$ may be constructed exactly once for each state $q$ of $\autm$. Hence the theorem.
\qed
\end{proof}

\begin{remark}
If we do not consider the DAG representation of \duds\/ formulas, then we must note that the size of the language-equivalent \duds\/ formula is exponential in the size of the original \potdfa.
\end{remark}

\section{Interval Temporal Logic \uitlpm}
\label{sec:uitlpm}
The interval logic \uitl\/ (\cite{LPS08}) has the unambiguous chop modalities which deterministically chop at the first and last occurrence of a letter $a$ within the interval. We enrich this logic with unambiguous modalities which chop beyond the interval boundaries in either direction. We call this logic \uitlpm. In this section, we introduce the logic \uitlpm\/  and show that it is no more expressive than \uitl\/, by giving an effective conversion from \uitlpm\/ formulas to their corresponding language-equivalent \tlxy\/ formula. The conversion is similar to the conversion from \uitl\/ to \tlxy\/, as given in \cite{DKL10}.

\subsection{\uitlpm: Syntax and Semantics}
The syntax and semantics of \uitlpm\/ are as follows:
\[
\begin{array}{l}
\top ~\mid~ a ~\mid~ \pti ~\mid~ \mathit{unit} ~\mid~ BP\phi ~\mid~ EP\phi ~\mid~
D_1 \firsta D_2 ~\mid~ D_1 \lasta D_2 ~\mid~ D_1 \firstp{a} D_2 ~\mid~D_1 \lastm{a} D_2 ~\mid ~\\
\succr D_1 ~\mid~ \predr D_1 ~\mid~ 
\succp D_1 ~\mid~ \predm D_1 ~\mid~ 
D_1 \lor D_2 ~\mid~ \neg D 
\end{array}
\]

Let $w$ be a nonempty finite word over $\Sigma$ 
and let $dom(w) = \{ 1, \ldots, \#w \}$ 
be the set of positions. Let 
$INTV(w) = \{ [i,j] ~\mid~ i,j \in dom(w), i \leq j \}~\cup~ \{\bot\}$ be the set of 
intervals over $w$, where $\bot$ is a special symbol to denote an undefined interval. For an interval $I$, let $l(I)$ and $r(I)$ denote the left and right endpoints of $I$. Further, if $I=\bot$, then $l(I)=r(I)=\bot$. The satisfaction of a formula $D$ is defined
over intervals of a word model $w$ as follows. \\
\[\begin{array}{l}
w,[i,j] \models \top \rmiff [i,j]\in INTV(w) \rmand [i,j]\neq\bot\\
w,[i,j]\models \pti \rmiff i=j\\
w,[i,j]\models \mathit{unit} \rmiff j=i+1\\
w,[i,j]\models BP\phi \rmiff w,[i,i]\models\phi\\
w,[i,j]\models EP\phi \rmiff w,[j,j]\models\phi
\end{array}
\]
\[\begin{array}{l}
w,[i,j] \models D_1 \firsta D_2 \rmiff \rmsome k: i \leq k \leq j \st ~~  
w[k]=a \rmand \\
\hspace*{1cm}  (\rmall m: i \leq m < k \st w[m] \neq a) \rmand \\ 
\hspace*{1cm}    w,[i,k] \models D_1  \rmand w,[k,j] \models D_2  \\
w,[i,j] \models D_1 \lasta D_2 \rmiff  \rmsome k: i \leq k \leq j \st ~~  
w[k]=a \rmand  \\ 
\hspace*{1cm} (\rmall m: k < m \leq j \st w[m] \neq a) \rmand  \\
\hspace*{1cm}     w,[i,k] \models D_1  \rmand w,[k,j] \models D_2  \\
w,[i,j] \models D_1 \firstp{a} D_2 \rmiff \rmsome k: k \geq j \st ~~  
w[k]=a \rmand \\
\hspace*{1cm}  (\rmall m: i \leq m < k \st w[m] \neq a) \rmand \\ 
\hspace*{1cm}    w,[i,k] \models D_1  \rmand w,[j,k] \models D_2  \\
w,[i,j] \models D_1 \lastm{a} D_2 \rmiff  \rmsome k: k \leq i \st ~~  
w[k]=a \rmand  \\ 
\hspace*{1cm} (\rmall m: k < m \leq j \st w[m] \neq a) \rmand  \\
\hspace*{1cm}     w,[k,i] \models D_1  \rmand w,[k,j] \models D_2  \\
w,[i,j] \models \succr D_1 \rmiff i< j \rmand w,[i+1,j]\models D_1\\
w,[i,j] \models \predr D_1 \rmiff i< j \rmand w,[i,j-1]\models D_1\\
w,[i,j] \models \succp D_1 \rmiff j<\#w \rmand w,[i,j+1]\models D_1\\
w,[i,j] \models \predm D_1 \rmiff i>1 \rmand w,[i-1,j]\models D_1\\
\end{array}
\]
The language $\mathcal L(\phi)$ of a \uitl\/ formula $\phi$ iff is given by $\mathcal L(\phi) = \{w ~ \mid ~ w,[1,\#w] \models \phi\}$. We may derive ``ceiling'' operators which assert the invariance as follows.
\begin{itemize}
\item $\aloo{A} ~\equiv~ \pti ~\lor~ \mathit{unit} ~\lor~ \neg\bigvee\limits_{b\not\in A}(\succr\predr(\top F_b \top))$\\
Hence, $w,[i,j]\models \aloo{A}$ if and only if $\forall i<k<j ~.~ w(k)\in A$.
\item $\aloc{A} ~\equiv~ \pti ~\lor~ \neg\bigvee\limits_{b\not\in A}(\succr(\top F_b \top))$\\
Hence, $w,[i,j]\models \aloc{A}$ if and only if $\forall i<k\leq j ~.~ w(k)\in A$.
\item $\alco{A} ~\equiv~ \pti ~\lor~ \neg\bigvee\limits_{b\not\in A}(\predr(\top F_b \top))$\\
Hence, $w,[i,j]\models \alco{A}$ if and only if $\forall i\leq k<j ~.~ w(k)\in A$.
\item $\alcc{A} ~\equiv~ \neg\bigvee\limits_{b\not\in A}(\top F_b \top)$\\
Hence, $w,[i,j]\models \alcc{A}$ if and only if $\forall i\leq k\leq j ~.~ w(k)\in A$.
\end{itemize}

\begin{example}
The language given in Example \ref{exam:potdfa} may be given by the \uitlpm\/ formula $\top L_a~(\aloo{\Sigma\setminus\{b\}} ~F_d\top)$.
\end{example}

\paragraph{\uitlpm\/ and Unique Parsing}
\uitlpm\/ is a deterministic logic and the property of \textit{Unique Parsing} holds for its subformulas. Hence, for every \uitlpm\/ subformula $\psi$, and any word $w$, there is a unique interval $\intv(\psi)$ within which it is evaluated. Further, for any ``chop'' operator ($\firsta,\lasta,\firstp{a},\lastm{a},\succr,\predr,\succp,\predm$), there is a unique chop position $\cpos(\psi)$. If such an interval or chop position does not exist in the word, then they are equal to $\bot$. The $\intv(\psi)$ and $\cpos(\psi)$ for any subformula $\psi$ depend on its context and may be inductively defined. (See \cite{LPS08} for similar such definition for the sublogic $\uitl$).

\subsection{From \duds\/ to \uitlpm}
Given a \duds\/ formula $\phi$, we shall construct a \uitlpm\/ formulas $BTrans(\phi)$ and $ETrans(\phi)$ having the following property. 
\begin{lemma}
\label{lem:dudsuitlpm}
Given a \duds\/ formula $\phi$, we may construct \uitlpm\/ formulas $BTrans(\phi)$ and $ETrans(\phi)$ such that for any word $w\in\Sigma^+$ and any interval $[i,j]$ in $w$
\begin{itemize}
\item $w,[i,j]\models BTrans(\phi)$ iff $w,i\models \phi$ 
\item $w,[i,j]\models ETrans(\phi)$ iff $w,j\models \phi$ 
\end{itemize}
The translation takes polynomial time.
\end{lemma}
\begin{proof}
The formulas $BTrans$ and $ETrans$ may be constructed by bottom-up induction using the following rules.
\begin{itemize}
\item $BTrans(a) ~=~ BP ~(\pti\firsta\top)$
\item $BTrans(\phi_1\lor\phi_2) ~=~ BTrans(\phi_1)\lor BTrans(\phi_2)$
\item $BTrans(\neg\phi) ~=~ \neg BTrans(\phi)$
\item $BTrans(A\detu_b\phi) ~=~ BP\succp\succr [~(\lcceil A\rceil) ~ \firstp{b} ~ETrans(\phi)]$
\item $BTrans(A\dets_b\phi) ~=~ BP\predm\predr [~(\lceil A\rcceil) ~ \lastm{b} ~BTrans(\phi)]$
\item $ETrans(a) ~=~ EP ~(\top\lasta\pti)$
\item $ETrans(\phi_1\lor\phi_2) ~=~ ETrans(\phi_1)\lor ETrans(\phi_2)$
\item $ETrans(\neg\phi) ~=~ \neg ETrans(\phi)$
\item $ETrans(A\detu_b\phi) ~=~ EP\succp\succr [~(\lcceil A\rceil) ~ \firstp{b} ~ETrans(\phi)]$
\item $ETrans(A\dets_b\phi) ~=~ EP\predm\predr [~(\lceil A\rcceil) ~ \lastm{b} ~BTrans(\phi)]$
\end{itemize}
The correctness of the above construction may be inferred from the semantics of the logics. For example, consider the formula $BTrans(A\detu_b\phi)$. Let us assume $ETrans(\phi)$ has been appropriately constructed so as to satisfy the lemma. Then for any word $w\in\Sigma^+$ and any interval $[i,j]$ of $w$, \\
$w,[i,j]\models BTrans(A\detu_b\phi)$\\
iff $w,[i,j]\models ~BP\succp\succr [~(\lcceil A\rceil) ~ \firstp{b} ~ETrans(\phi)]$\\
iff $w,[i,i] \models ~\succp\succr [~(\lcceil A\rceil) ~ \firstp{b} ~ETrans(\phi)]$\\
iff $w,[i+1,i+1]\models ~[~(\lcceil A\rceil) ~ \firstp{b} ~ETrans(\phi)] $\\
iff $\exists k\geq (i+1) ~.~ w(k)=b \land ~\forall (i+1)\leq m<k ~.w(m)\in A\setminus \{b\} ~\land~ $\\
\hspace*{1cm} $w,[i+1,k]\models ETrans(\phi)$\\
iff $w,i\models A\detu_b\phi$
\qed
\end{proof}

From the above construction, we infer that for every \duds\/ formula, we may construct a language-equivalent \uitlpm\/ formula whose size is linear in the size of the \duds\/ formula. Clearly, the time time taken for the construction is also polynomial.

\subsection{\uitlpm\/ to \tlxy}
In \cite{LPS08}, we exploited the interval-nesting structure of \uitl\/ formulas to give a reduction from \uitl\/ to \potdfa\/. However such a nesting structure is absent in the case of \uitlpm and the translation presented in \cite{LPS08} can not be extended to \uitlpm. The reduction from\uitlpm\/ formulas to \potdfa\/ is factored via \tlxy\/. This translation is interesting and it uses the concept of ranker directionality. 

\begin{theorem} \label{theo:uitlpm}
Given any \uitlpm\/ formula $\phi$ of size $n$, we can construct in polynomial time a language-equivalent \tlxy\/ formula
$\mathit{Trans}(\phi)$, whose size is $O(n^2)$. Hence, satisfiability of 
\uitlpm\/ is NP-complete. 
\end{theorem}
The  construction of $\mathit{Trans}(\phi)$ requires some auxiliary definitions.
For every \uitlpm\/ subformula $\psi$ of $\phi$, we define \stls\/ $\mathit{LIntv}(\psi)$ and $\mathit{RIntv}(\psi)$, such that Lemma \ref{lem:uitlpm} holds. $\mathit{LIntv}(\psi)$ and $\mathit{RIntv}(\psi)$ are \stls\ which accept at the left and right ends of the unique interval $\intv(\psi)$ respectively.
\begin{lemma}
\label{lem:uitlpm}
Given a \uitlpm\/ subformula $\psi$ of a formula $\phi$, and any $w\in\Sigma^+$ such that $\intv(\psi),\cpos(\psi)\neq\bot$, 
\begin{itemize}
 \item $\lpos(\mathit{LIntv}(\psi)) = l(\intv(\psi))$
 \item $\lpos(\mathit{RIntv}(\psi)) = r(\intv(\psi))$
\end{itemize}
\end{lemma}
The required formulas $\mathit{LIntv}(\psi), \mathit{RIntv}(\psi)$ may be constructed by induction on the depth of occurrence of the subformula $\psi$ as below. The correctness of these formulas  is apparent from the semantics of \uitlpm\/ formulas, and we omit the detailed proof.\\
\begin{itemize}
 \item If $\psi=\phi$, then $\mathit{LIntv}(\psi) = SP \top$, $\mathit{Rintv}(\psi) = EP \top$
\item If $\psi=BP ~D_1$ then\\
$\mathit{LIntv}(D_1)=\mathit{RIntv}(D_1)=\mathit{LIntv}(\psi)$
\item If $\psi=EP ~D_1$ then\\
$\mathit{LIntv}(D_1)=\mathit{RIntv}(D_1)=\mathit{RIntv}(\psi)$
 \item If $\psi = D_1\firsta D_2$ then\\
 $\mathit{LIntv}(D_1) = \mathit{LIntv}(\psi) $, $\mathit{Rintv}(D_1) = \mathit{LIntv}(\psi)~;~ \weakx{a}\top$,\\
 $\mathit{LIntv}(D_2) = \mathit{LIntv}(\psi) ~;~ \weakx{a}\top$, $\mathit{Rintv}(D_2) = \mathit{RIntv}(\psi)$
\item If $\psi=D_1\firstp{a}D_2$ then\\
 $\mathit{LIntv}(D_1) = \mathit{LIntv}(\psi) $, $\mathit{Rintv}(D_1) = \mathit{RIntv}(\psi)~;~ \weakx{a}\top$,\\
 $\mathit{LIntv}(D_2) =\mathit{RIntv}(\psi) $, $\mathit{Rintv}(D_2) = \mathit{RIntv}(\psi) ~;~ \weakx{a}\top$
\item If $\psi=D_1\lasta D_2$ then\\
$\mathit{LIntv}(D_1) = \mathit{LIntv}(\psi) $, $\mathit{Rintv}(D_1) = \mathit{RIntv}(\psi)~;~ \weaky{a}\top$,\\
 $\mathit{LIntv}(D_2) = \mathit{RIntv}(\psi) ~;~ \weaky{a}\top$, $\mathit{Rintv}(D_2) = \mathit{RIntv}(\psi)$
\item If $\psi=D_1\lastm{a}D_2$ then\\
$\mathit{LIntv}(D_1) = \mathit{LIntv}(\psi) ~;~ \weaky{a}\top $, $\mathit{Rintv}(D_1) = \mathit{LIntv}(\psi)$,\\
 $\mathit{LIntv}(D_2) = \mathit{LIntv}(\psi) ~;~ \weaky{a}\top$, $\mathit{Rintv}(D_2) = \mathit{RIntv}(\psi)$
\item If $\psi = \succr D_1$ then \\
$\mathit{LIntv}(D_1) = \mathit{LIntv}(\psi)~;~ \xunit\top$, $\mathit{RIntv}(D_1)=\mathit{RIntv}(\psi)$
\item If $\psi = \succp D_1$ then \\
$\mathit{LIntv}(D_1) = \mathit{LIntv}(\psi)$, $\mathit{RIntv}(D_1)=\mathit{RIntv}(\psi) ~;~ \xunit\top$
\item If $\psi = \predr D_1$ then \\
$\mathit{LIntv}(D_1) = \mathit{LIntv}(\psi)$, $\mathit{RIntv}(D_1)=\mathit{RIntv}(\psi) ~;~ \yunit\top$
\item If $\psi = \predm D_1$ then \\
$\mathit{LIntv}(D_1) = \mathit{LIntv}(\psi) ~;~ \yunit\top$, $\mathit{RIntv}(D_1)=\mathit{RIntv}(\psi) $
\end{itemize}

We can now construct, for any subformula $\psi$ of $\phi$, a corresponding \tlxy\/ formula $\mathit{Trans}(\psi)$. The conversion uses the following inductive rules. 
Then, it is easy to see that $\mathit{Trans}(\psi)$ is language equivalent to $\phi$
(see \cite{SShah12} for proof). 
\begin{itemize}
\item If $\psi= BP ~D_1$ or $EP ~D_1$ then $\mathit{Trans}(\psi) = \mathit{Trans}(D_1)$
 \item If $\psi = D_1\firsta D_2$, then $\mathit{Trans}(\psi)= [(~\mathit{LIntv}(\psi);\weakx{a}\top~)~;~\pfleq(\mathit{RIntv}(\psi))] \land \mathit{Trans}(D_1)\land \mathit{Trans}(D_2)$
 \item If $\psi = D_1\lasta D_2$, then $\mathit{Trans}(\psi)= [(~\mathit{RIntv}(\psi);\weaky{a}\top~)~;~\pfgeq(\mathit{LIntv}(\psi))] \land \mathit{Trans}(D_1)\land \mathit{Trans}(D_2)$
 \item If $\psi = D_1\firstp{a} D_2$, then $\mathit{Trans}(\psi)= [(~\mathit{LIntv}(\psi);\weakx{a}\top~)~;~\pfgeq(\mathit{RIntv}(\psi))] \land \mathit{Trans}(D_1)\land \mathit{Trans}(D_2)$
 \item If $\psi = D_1\lastm{a} D_2$, then $\mathit{Trans}(\psi)= [(~\mathit{RIntv}(\psi);\weaky{a}\top~)~;~\pfleq(\mathit{LIntv}(\psi))] \land \mathit{Trans}(D_1)\land \mathit{Trans}(D_2)$
\item If $\psi= \succr D_1$, then $\mathit{Trans}(\psi)=[(\mathit{LIntv}(\psi);\xunit\top) ~;~ \pfleq(\mathit{RIntv}(\psi))] ~\land~ \mathit{Trans}(D_1)$
\item If $\psi= \predr D_1$, then $\mathit{Trans}(\psi)=[(\mathit{RIntv}(\psi);\yunit\top) ~;~ \pfgeq(\mathit{LIntv}(\psi))] ~\land~ \mathit{Trans}(D_1)$
\item If $\psi= \succp D_1$, then $\mathit{Trans}(\psi)=[(\mathit{RIntv}(\psi);\xunit\top)] ~\land~ \mathit{Trans}(D_1)$
\item If $\psi= \predm D_1$, then $\mathit{Trans}(\psi)=[(\mathit{LIntv}(\psi);\yunit\top)] ~\land~ \mathit{Trans}(D_1)$
 \item $\mathit{Trans}(D_1\lor D_2) = \mathit{Trans}(D_1)\lor \mathit{Trans}(D_2)$
 \item $\mathit{Trans}(\neg D_1) = \neg \mathit{Trans}(D_1)$
\end{itemize}

\section{Bridging the Gap: From Deterministic to Non-deterministic Logics}
\utl\/ is the unary fragment of the well known Linear Temporal Logic, with the unary modalities $\fut$ (\textit{future}) and $\past$ (\textit{past}) and the boolean operators. \utl\/ was studied by Etessami, Vardi and Wilke \cite{EVW02} who showed that it belongs to the language class $UL$. They also
showed that the satisfiability of $\utl$ is \np-complete by giving a small model property for $\utl$ formulas. We derive here, an explicit translation from \utl\/ formulas to language-equivalent \tlxy\/ formulas and analyse its size. This will not only allow us to construct an equivalent \potdfa\/ for the \utl\/ formula but also give an alternative proof for their \np-complete satisfiability.

Let $a\in\Sigma$. The syntax and semantics of \utl\/ formulas is as follows.
\[
a ~\mid~ \fut\phi  ~\mid~ \past\phi ~\mid~ \phi\lor\phi  ~\mid~ \neg\phi
\]
Given any word $w\in\Sigma^*$ and $i\in dom(w)$, \utl\/ formulas are interpret over words as follows.
\begin{center}
\begin{tabular}{rcl}
$w,i\models a$ & iff & $w(i)=a$\\
$w,i\models \fut\phi$ & iff & $\exists j>i ~.~ w,j\models \phi$\\
$w,i\models \past\phi$ & iff & $\exists j<i ~.~ w,j\models \phi$\\
\end{tabular}
\end{center}
The boolean operators have their usual meaning. Given a \utl\/ formula $\phi$, the language defined by $\phi$ is given by $\mathcal L(\phi) ~=~ \{w ~\mid~ w,1\models\phi\}$.

\paragraph{Modal subformulas and Boolean subformulas: }
Every modal subformula $\psi=\fut\phi$ or $\psi=\past\phi$ is such that $\phi=\mathscr B(\psi_i)$, where each $\psi_i$ is in turn either a modal subformula or an atomic formula and $\mathscr B$ is a boolean function. We shall use $\psi$ to denote modal subformulas and $\phi$ to denote the boolean formulas. $\psi$ is a $\fut$-type or $\past$-type formula depending on the outer modality of $\psi$. For any subformula $\xi$, let $Sform(\xi)$ denote the set of modal subformulas of $\xi$ (excluding $\xi$) and $\mathit{Iform}(\xi)\subseteq Sform(\xi)$ denote the set of immediate modal subformulas of $\xi$.

\paragraph{Validity of modal subformulas\\}
Given a word $w$ and a modal subformula $\psi$, $\psi$ is said to be \textit{defined} in $w$ if $\exists i\in dom(w) ~.~ w,i\models \psi$. We call the last position (in case $\psi$ is $\fut$-type) or the first position (in case $\psi$ is $\past$-type) in $w$ where $\psi$ holds, as the \emph{defining position} of $\psi$ in $w$. This is denoted as $d\pos(\psi)$. In case $\psi$ is not defined in $w$, then its defining position does not exist, and is equal to $\bot$. Thus $d\pos(\psi)\in ~dom(w)\cup \{\bot\}$.

\subsection{\utl\/ to \tlxy}
Representing the non-deterministic $\fut$ and $\past$ operators of $\utl$ in deterministic $\tlxy$ is challenging. A critical property of the unary modalities is the following. In any given word $w$ if a modal subformula of the form $\fut \phi$ is defined in $w$, then it holds at exactly all positions within an interval $[1,i-1]$, where $i$ is the last position in $w$ where $\phi$ is defined. Similarly, if a modal subformula of the form $\past\phi$ is defined in $w$ then it holds exactly at all positions within an interval $[j+1,\#w]$ where $j$ is the first position in $w$ where $\phi$ is defined. 

The following proposition relates the defining position of modal formulas of the form $\fut\phi$ or $\past\phi$ to the first or last position where $\phi$ is defined. Its correctness may be directly inferred from the semantics of $\fut$ and $\past$ operators. 

\begin{proposition}
\label{prop:psi}
\begin{itemize}
\item If $\psi=\fut\phi$ and $i$ is the last position in $w$ where $\phi$ holds then 
\begin{itemize}
\item $d\pos(\psi) = i-1$ (if $i>1$)
\item $\forall j\leq d\pos(\psi) ~.~ w,j\models \psi$
\end{itemize}
\item If $\psi=\past\phi$ and $i$ is the first position in $w$ where $\phi$ holds then 
\begin{itemize}
\item $d\pos(\psi) = i+1$ (if $i<\#w$)
\item $\forall j\geq d\pos(\psi) ~.~ w,j\models \psi$
\end{itemize}
\end{itemize}
\end{proposition}

\paragraph{Region partitioning\\}
Our translation from \utl\/ formulas to \tlxy\/ formulas relies on the following key observation, which is closely related to Proposition \ref{prop:psi}.
\begin{quote}
In the evaluation of a \utl\/ formula over a word $w$, it is sufficient to determine the relative positioning of the $d\pos$ positions of the modal subformulas and the occurrence of letters (of the alphabet) between them.
\end{quote}

Consider a set of modal subformulas $\kappa=\{\psi_1\cdots \psi_n\}$ and a word $w$ such that every $\psi_i$ is defined in $w$. The defining positions of $\psi_i$ partition $w$ into ``regions'', such that each region is either a defining position of one or more $\psi_i$ (called a formula region or F-region), or the region lies strictly between two consecutive defining positions (called an Intermediate region or I-region). While each F-region consists of exactly one position in $w$, an $I$-region is a subword of length 0 or more. The region partitioning comprises of alternating I and F-regions, along with a specification of the subset of the alphabet that occurs within these regions, as well as their order of first / last appearances within each region.   
\begin{example}
\label{exam:region}
Consider a set of modal formulas $\kappa = \{\psi_1,\psi_2,\psi_3,\psi_4\}$ that are defined in a word $w$. The orientation of their defining positions is as depicted in Figure \ref{fig:regionpart}. We have $d\pos(\psi_1)=d\pos(\psi_2) ~>1$ and $d\pos(\psi_3)=\#w$. The region partitioning of $\kappa$ in $w$ is given as $r_1,r_2,r_3,r_4,r_5,r_6$, where $r_1,r_3,r_5$ are I-regions and $r_2,r_4,r_6$ are F-regions. Further, if the region $r_3$ corresponds to the subword $s=aabcddcbcdac$ then its corresponding alphabet is $\{a,b,c,d\}$ and its order of occurrence is $a,b,c,d$ and $c,a,d,b$ from the left and right, respectively.
\end{example}
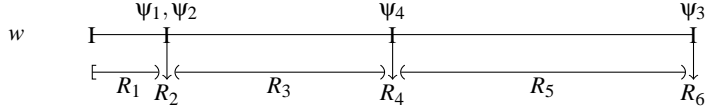
\begin{figure}
\begin{tikzpicture}
\draw (1,3) node{I}--(2,3) node{I} --(5,3) node{I}-- (9,3)node{I};
\draw (0,3) node{$w$};
\draw (2,3.3) node{$\psi_1,\psi_2$}; \draw (5,3.3) node{$\psi_4$}; \draw (9,3.3) node{$\psi_3$}; 
\draw [[-)](1,2.5)--(1.9,2.5); \draw[(-)](2.1,2.5)--(4.9,2.5); \draw [(-)](5.1,2.5)--(8.9,2.5);
\draw [->](2,3)--(2,2.4); \draw [->](5,3)--(5,2.4); \draw [->] (9,3)--(9,2.4);
\draw (1.5,2.3)node{$R_1$}; \draw (2,2.2) node{$R_2$};\draw (3.5,2.3) node{$R_3$}; \draw (5,2.2) node{$R_4$};\draw(7,2.3)node{$R_5$}; \draw (9,2.2) node{$R_6$};
\end{tikzpicture}
\caption{Region partitioning of $\kappa$ in $w$}
\label{fig:regionpart}
\end{figure}

\paragraph{Region Templates\\}
For a given set of modal formulas $\kappa$, there are only a finite number of possible relative orderings of defining positions of modal formulas in $\kappa$. We shall call each such ordering, along with the specification of letter occurrences between them as a \emph{region template}. Hence, the set of all possible region templates partitions the set of all words (in which all formulas of $\kappa$ are defined) into a finite number of equivalence classes. 

Formally, a region template $\mathscr R(\kappa)$ of a set of modal subformulas $\kappa=\{\psi_1 \cdots \psi_n\}$ is a tuple $(S,<_S,\tau,\alpha,\beta)$, where
\begin{itemize}
\item $S$ is a finite set of I-regions and F-regions.
\item $<_S$ is a strict total ordering on the set $S$ such that the I-regions and F-regions alternate.
\item $\tau: S\to 2^\kappa$ is a function which maps the F-regions to the set of subformulas whose defining position corresponds to that region. For every I-region $r$, $\tau(r)=\emptyset$ and for every F-region $r$, $\tau(r) \neq \emptyset$. Further, for every $\psi_i\in\kappa$, there exists a unique F-region $r\in S$ such that $\psi_i\in \tau(r)$, and this unique region is denoted as $reg(\psi_i)$.
\item $\alpha: S\to 2^\Sigma$ maps every region to the subset of letters. Note that for every F-region $r$, $\alpha(r)$ is a singleton. 
\item $\beta$ is a function which maps each region $r$ to a pair of ordering relations $<^L,<^R$ over the set $\alpha(r)$. $<^L$ and $<^R$ are strict total orders. 
\end{itemize}

Given a region template $\mathscr R(\kappa)=(S,<_S,\tau,\alpha,\beta)$ and a word $w\in \Sigma^+$ such that each $\psi_i\in \kappa$ is defined in $w$, we say that \emph{$\mathscr R(\kappa)$ is the (unique) region template of $w$ for $\kappa$} if there exists a partitioning $\mathit{Part}$ of $w$ such that there exists a bijection $\mathit{Equiv}~:~S\to \mathit{Part}$ which preserves the ordering relation $<_S$ and satisfies the following conditions
\begin{itemize}
\item For all F regions $r\in S$, the corresponding subword $p\in Part$ is a subword with a single position $i\in dom(w)$ such that $\forall \zeta\in \tau(r) ~.~ d\pos(\zeta)=i$.
\item For all regions $r\in S$, the corresponding subword $p\in \mathit{Part}$ is such that $\forall a\in \Sigma$. $a\in \alpha(r)$ if and only if $a$ occurs in $p$.
\item For all regions $r\in S$, the corresponding  subword $p\in \mathit{Part}$ is such that the ordering relations $<^L$ and $<^R$ exactly correspond to the ordering of first appearance of the letters in $p$ from the left and right respectively.
\end{itemize} 
Consider the region partitioning of the word $w$ in Example \ref{exam:region} (Figure \ref{fig:regionpart}) and region template $\mathscr R$ given by the sequence $S=\{r_1,r_2,r_3,r_4,r_5,r_6\}$, with $\tau(r_2)=\{\psi_1,\psi_2), ~\tau(r_4)=\psi_4, ~\tau(r_6)=\psi_3$, $tau(r_1)=\tau(r_3)=\tau(r_5)=\emptyset$, and the region $r_3$ is such that $\alpha(r_3)=\{a,b,c,d\}$, $a~<_L~b~<_L~c<_L~d$ and $c~<_R~a~<_R~d~<_R~b$ (and similarly for other regions as well). Then we may say that $\mathscr R$ is the region template of $w$ for $\{\psi_1,\psi_2,\psi_3,\psi_4\}$.

The proposition below may be inferred from the following property: Given a word $w$ and a modal formula $\psi$ that is defined in $w$, there exists a unique defining position of $\psi$ in $w$.
\begin{proposition}
\label{prop:unireg}
Given a set of modal subformulas $\kappa$ and any $w\in\Sigma^+$ such that every formula in $\kappa$ is defined in $w$, there exists a \emph{unique} region template $\mathscr R$ such that $\mathscr R$ is the region template of $w$ for $\kappa$. 
\end{proposition}

In the remainder of the section, we shall often refer to a region $r$ in a word $w$, to mean the partition in the $w$ which corresponds to the $r$ (that is given by the equivalence $\mathit{Equiv}$).

\paragraph{Parameters $\Delta$ and $\theta$\\}
Let $\Phi$ be a \utl\/ formula. We shall construct a \tlxy\/ formula $\mathit{Trans}(\Phi)$ that is language-equivalent to $\Phi$. 
For the top-level formula $\Phi$, we define \emph{parameters} $\Delta$ and $\theta$ of $\Phi$ as follows. $\Delta~\subseteq ~Sform(\Phi)$ is a subset of the set of modal subformulas of $\Phi$. $\theta$ is a function which maps each modal subformula $\psi$ of $\Phi$ to a region template over the set $Iform(\psi)\cap \Delta$.

\begin{definition}
\label{def:theta}
Given a word $w\in\Sigma^*$, $w$ is said to \emph{conform to} parameters $\Delta$ and $\theta$ if $\Delta$ is exactly the subset of modal subformulas of $\Phi$ which are defined in $w$ and for every $\psi\in Sform(\Phi)$, $\theta(\psi)$ is the region template of $w$ for the set $Iform(\psi)\cap \Delta$. 
\end{definition}

\paragraph{Evaluating Boolean Formulas\\}
Fix parameters $\Delta$ and $\theta$ for $\Phi$. For a boolean subformula $\phi$ of $\Phi$, we may construct a set $Def^{\Delta,\theta}(\phi)$ which is a set of pairs $\{(r,A)\}$ such that $r\in S$ and $A\subseteq \alpha(r)$ (and $A\neq\emptyset$). The idea behind the construction of $Def(\phi)$ is to identify exactly the positions where $\phi$ will hold. The validity of $\phi=\mathscr B(\psi_j)$ at a position $i$ in a word depends on the following:
\begin{itemize}
\item the relative positioning of $i$ with respect to the defining positions of the modal subformulas in $\{\psi_j\}$, and hence the region (in the region partitioning of $Iform(\phi)$) to which $i$ belongs.
\item the letter $w(i)$ at the position $i$- to infer the validity of the atomic formulas in $\{\psi_j\}$.
\end{itemize}
Hence, the set $Def^{\Delta,\theta}(\phi)$ exactly indicates in terms of $(r,A)$ pairs, the positions in a word where $(\phi)$ will hold. The construction of $Def^{\Delta,\theta}(\phi)$ is formulated in the lemma below.

\begin{lemma}
\label{lem:def}
Given $\Delta,\theta$ of a formula $\Phi$ and a boolean subformula $\phi=\mathscr B(\zeta_j)$ of $\Phi$, the set $Def^{\Delta,\theta}(\phi)$ may be constructed such that for all words $w$ that \emph{conform to} $\Delta,\theta$, and for all $i\in dom(w)$, $w,i\models \phi$ if and only if $\exists (r,A)\in Def^{\Delta,\theta}(\phi)$ such that $i\in r$ and $w(i)\in A$.
\end{lemma}
\begin{proof}
Consider a modal subformula $\psi=\fut\phi$ (or alternatively $\past\phi$) such that $\phi=\mathscr B(\zeta_j)$, where each $\zeta_j$ is in turn a modal formula or an atomic formula. Let $\theta(\psi)=\mathscr R=(S,<_S,\tau,\alpha,\beta)$. The set $Def^{\Delta,\theta}(\phi)$ may be constructed by structural induction on $\phi$. 

\begin{itemize}
\item If $\phi=a$, then $Def^{\Delta,\theta}(\phi) = \{(r,\{a\}) ~\mid~ r\in S~\land~ a\in \alpha(r)\}$
\item If $\phi=\phi_1\land \phi_2$ then $Def^{\Delta,\theta}(\phi) = \{(r,A_1\cap A_2) ~\mid~ (r,A_1)\in Def^{\Delta,\theta}(\phi_1) ~\land~ (r,A_2)\in Def^{\Delta,\theta}(\phi_2)\ ~\land~ A_1\cap A_2\neq \emptyset\}$
\item If $\phi=\neg \phi_1$ then $Def^{\Delta,\theta}(\phi) = \{(r,\Sigma\setminus A) ~\mid~ (r,A)\in Def^{\Delta,\theta}(\phi_1)\ ~\land~ A\neq\Sigma\}$
\item If $\phi=\phi_1\lor \phi_2$ then $Def^{\Delta,\theta}(\phi) = \{(r,A_1\cup A_2) ~\mid~ (r,A_1)\in Def^{\Delta,\theta}(\phi_1) ~\land~ (r,A_2)\in Def^{\Delta,\theta}(\phi_2)\}$
\item If $\phi=\zeta$ where $\zeta=\fut(\phi')$ then \\
$Def^{\Delta,\theta}(\phi) = \{(r,\alpha(r)) ~\mid~ r\leq_{S} reg(\zeta)~\land~ \alpha(r) \neq\emptyset\}$, if $\zeta\in \Delta$\\
$Def^{\Delta,\theta}(\phi) = \emptyset$, if $\zeta\not\in \Delta$\\
\item If $\phi=\zeta$ where $\zeta=\past(\phi')$ then \\
$Def^{\Delta,\theta}(\phi) = \{(r,\alpha(r)) ~\mid~ r\geq_{S} reg(\zeta) ~\land~ \alpha(r)\neq\emptyset\}$, if $\zeta\in \Delta$\\
$Def^{\Delta,\theta}(\phi) = \emptyset$, if $\zeta\not\in \Delta$\\
\end{itemize}
The correctness of the above construction may be deduced by induction on the structure of $\phi$ using the semantics of the logic \utl, Proposition \ref{prop:psi} and the fact that $w$ \emph{conforms to} $\Delta,\theta$. The atomic and boolean cases are straightforward. Consider the interesting case of $\phi=\fut\phi' (=\zeta)$. From Proposition \ref{prop:psi}, we know that $\phi$ holds true at all positions that are at or before $d\pos(\zeta)$. Hence for any $w$, since $w$ conforms to $\Delta,\theta$, we know that $d\pos(\zeta)=(reg(\zeta))$. Therefore we know that $\phi(\zeta)$ holds at all regions at or before $reg(\zeta)$.  
\end{proof}

\paragraph{Constructing the ranker for $\psi$\\}
Using a bottom-up induction, for every modal subformula $\psi\in \Delta$, we may construct a ranker $D^{\Delta,\theta}(\psi)$ such that for all words $w$ which \emph{conform to} $\Delta,\theta$, the ranker $D^{\Delta,\theta}(\psi)$ accepts at $d\pos(\psi)$. 

Given the set $Def^{\Delta,\theta}(\phi)$, we may construct the ranker $D^{\Delta,\theta}(\psi)$ for the modal subformula $\psi=\fut\phi$ or $\past\phi$ as follows. Let $\undef$ be a special ranker which does not accept on any word. If $Def^{\Delta,\theta}(\phi)=\emptyset$, then $D^{\Delta,\theta}(\psi)=\undef$. \\
Otherwise, if $Def^{\Delta,\theta}(\phi)$ is non-empty, then let $min(Def^{\Delta,\theta}(\phi),<_S)$\footnote{In general, given a set $A$ and a total ordering $<$ on $A$, let $min(A,<)$ and $max(A,<)$ be the minimal and maximal elements (respectively) of $A$ with respect to the ordering $<$.} and $max(Def^{\Delta,\theta}(\phi),leq_S)$ denote the minimal and maximal elements of $(Def^{\Delta,\theta})$ wrt the ordering $<_{S}$ of the regions.\footnote{From the construction of $Def^{\Delta,\theta}(\phi)$ it is apparent that for every region $R$, there is at most one element with $R$ in $Def^{\Delta,\theta}(\phi)$. }\\
If $\psi=\fut\phi$, then from Proposition \ref{prop:psi}, we know that $(Def^{\Delta,\theta})$ must accept at one position previous to the maximum position where $\phi$ holds. Such a ranker is constructed as follows:
\begin{itemize}
\item Case: If $max(Def^{\Delta,\theta}(\phi),<_S)=(r,A)$ such that $r$ is an F-region, then $\tau(r)\neq\emptyset$ and for some $\zeta$, $\zeta\in \tau(r)$, then \\
\[
D^{\Delta,\theta}(\psi)= D^{\Delta,\theta}(\zeta);\yunit\top\\
\]
\item Case: If $max(Def^{\Delta,\theta}(\phi),<_S)=(r,A)$, such that $\tau(r)=\emptyset$ (i.e. $r$ is an I-region) then
\begin{itemize}
\item If $r=max(S,<_S)$, then $r$ includes the last position in the word. Hence\\
\[
D^{\Delta,\theta}(\psi)=EP\weaky{p}\yunit\top\\
\]
where $p=min(A\cap \alpha(r), <^R)$.
\item If $r\neq max(S,<_S)$, then if $r'$ is the region subsequent to $r$, there exists $\zeta$ such that $reg(\zeta)=r'$. Then
\[
D^{\Delta,\theta}(\psi)=D^{\Delta,\theta}(\zeta);Y_{p}\yunit\top\\
\]
where $p=min(A\cap \alpha(r), <^R)$.
\end{itemize}
\end{itemize}
The ranker for the case of $\psi=\past\phi$ is symmetric to the above. 

The correctness of this construction is given by Lemma \ref{lem:para} part(ii).

\paragraph{Checking $\Delta$ and $\theta$\\}
We shall now give the formulas which ``check'' whether a given word \emph{conforms to} a given $\Delta$ and $\theta$. For convenience and ease of readability, we have dropped the superscript $\Delta,\theta$. \\

The formula $Dvalid$ checks if $\Delta$ holds for the given word.\\
\[
Dvalid(\Delta) ~=~ \bigwedge\limits_{\psi\in\Delta}(D(\psi)) ~~\land~~ \bigwedge\limits_{\psi\not\in\Delta}(\neg D(\psi))\\
\]
The formula $Tvalid$ checks for the correctness of $\theta$ by checking for each modal subformula $\psi$ whether $\theta(\psi)$ is the region template of the word, wrt the set $Iform(\psi)\cap\Delta$.
\[
Tvalid(\Delta,\theta) = \bigwedge\limits_{\psi\in Sform(\Phi)\cup\Phi} [Rvalid(\theta,\psi) ~\land~ Avalid(\theta,\psi) ~\land~ Bvalid(\theta,\psi)]
\]
In the above, if $\theta(\psi)= (S,<_S,\tau,\alpha,\beta)$ then $Rvalid(\theta,\psi)$ checks the consistency of $<_S$ and $\tau$. $Avalid(\theta,\psi)$ and $Bvalid(\theta,\psi)$ respectively check the correctness of $\alpha$ and $\beta$ in the given word. They are as given below. Assume that for each $\psi$, $\theta(\psi)= (S,<_S,\tau,\alpha,\beta)$ such that $r_1,\cdots r_{maxR\psi}$ is the enumeration of the regions in $S$ based on the ordering $<_S$. 

$RValid$ checks the validity of $\tau(r_i)$ for all the F-regions $r_i$ and also the relative ordering of the F-regions, which implicitly also verifies the ordering of I-regions that alternate with the F regions. While $TauChk(r_i)$ checks whether the rankers corresponding to every $\zeta\in\tau(r_i)$ accept at the same position, $OrdChk(r_i)$ checks the relative ordering of successive F-regions, using the rankers of the modal formulas that are contained in $\tau(r_i)$. These formulas are as given below.
\[
Rvalid(\theta,\psi) ~=~ \bigwedge\limits_{i\in \{1,\cdots maxR\psi\}}[\tau(r_i)\neq \emptyset ~\implies~ (TauChk(r_i) ~\land~ OrdChk(r_i))]
\]
\[
TauChk(r_i) = \bigwedge\limits_{\zeta,\xi\in \tau(r_i)}[D(\zeta);\pfleq(D(\xi)) ~\land~ D(\xi);\pfleq(D(\zeta))]
\]
\[
OrdChk(r_i) = D(\zeta);\pfless(D(\xi))
\]
where $\zeta\in\tau(r_i)$ and $\xi\in\tau(r_{i+2})$, (for $i\leq maxR\psi-2$)

The formula $Avalid$ checks the presence of the letters in $\alpha(r_i)$ within the region $r_i$, using $ChkLet(r_i)$ and at the same time, it checks for the absence of letters which are not in $\alpha(r_i)$. This is done using ranker-directionality formulas for rankers corresponding to F-regions. 
\[
Avalid(\theta,\psi) ~=~ \bigwedge\limits_{i\in \{1,\cdots maxR\psi\}}[ChkLet(r_i) ~\land ChkNot(r_i)]
\]
Case: $r_i$ is an I-region and $1~<~i~<~maxR\psi$. Let $\zeta\in tau(r_{i-1})$ and $\xi\in\tau_(r_{i+1})$. Then\\
\[
ChkLet(r_i) ~=~  \bigwedge\limits_{a\in \alpha(r_i)} [D(\zeta);~ X_a;~\pfless(D(\xi))]
\]
\[
ChkNot(r_i) ~=~ \bigwedge\limits_{a\not\in \alpha(r_i)}\neg[D(\zeta);~ X_a~;\pfless(D(\xi))]
\]
The other cases where $r_i$ is an I-region and it is either the first or last region, or if $r_i$ is an F-region, may be worked out similarly.

The formula $Bvalid$ checks for each region, the ordering of the letters within the region, from the left side (using $LOrdChk$) and from the right side (using $ROrdChk$).  
\[
Bvalid(\theta,\psi) ~=~ \bigwedge\limits_{i\in \{1,\cdots maxR\psi\}}[LOrdChk(r_i) ~\land~ ROrdChk(r_i)]
\]
If $r_i$ is an F-region then $\alpha(r_i)$ is a singleton. Hence the interesting case is when $r_i$ is an I-region. \\
Case: $r_i$ is an I-region and $1~<~i~<~maxR\psi$.  Let $\xi\in \tau(r_{i-1}), ~\zeta\in\tau(r_{i+1})$ and $\{b_1...b_m\}\in \alpha(r_i)$.
\[
LOrdChk ~=~ \bigwedge\limits_{j\in\{1...m\}}[D(\xi)X_{b_j} ;\pfless(D(\xi);X_{b_{j+1}}\top))]
\]
\[
ROrdChk ~=~ \bigwedge\limits_{j\in\{1...m\}}[D(\zeta)Y_{b_j} ;\pfgreat(D(\zeta);Y_{b_{j+1}}\top))]
\]
Other cases where $i=1$ or $i= maxR\psi$, may be worked out similarly. 

The following lemma asserts the correctness of the above validity-check formulas for the parameters and also the correctness of the ranker construction for the modal subformulas.

\begin{lemma}
\label{lem:para}
\begin{itemize}
\item[(i)] Given parameters $\Delta,\theta$ of $\Phi$, for all $w\in \Sigma^+$, $w$ \emph{conforms to} $\Delta,\theta$ if and only if
\begin{itemize}
\item $w\models Dvalid(\Delta)$ and
\item $w\models Tvalid(\theta)$
\end{itemize}
\item[(ii)] Given parameters $\Delta,\theta$ of $\Phi$ and a modal subformula $\psi$ of $\Phi$, for every $w\in\Sigma^+$ such that $w$ \emph{conforms to} $\Delta,\theta$, the ranker $D^{\Delta,\theta}(\psi)$ accepts at a position $i\in dom(w)$ if and only if $\psi$ is defined in $w$ and $d\pos(\psi)=i$.
\end{itemize}
\end{lemma}
\begin{proof}
Given a modal subformula $\psi$ of $\Phi$ such that $\psi=\fut/\past\phi$, let $\Delta_\phi$ and $\theta_\phi$ be the restrictions of $\Delta$ and $\theta$ to $\phi$. Therefore, $\Delta_\phi=\Delta\cap Sform(\phi)$ and $\theta_\phi$ is the restriction of the function $\theta$ to the domain $Sform(\psi)\cup\psi$. 

We shall prove the lemma by induction on the depth of the subformulas. Consider a modal subformula $\psi=\fut/\past (\phi)$ of $\Phi$ such that $\phi=\mathscr B(\zeta_i)$ where each $\zeta_i$ is a modal subformula or atomic formula. 
\begin{itemize}
\item Base Case:\\
If $Iform(\psi) = \emptyset$ then $\phi$ is a boolean combination of atomic formulas. Hence $\Delta_\phi=\emptyset$ and $\theta_\phi(\psi)=\mathscr R$. Here, the only possible region set of $\mathscr R$ is one which consists of a single region $r$ such that $\tau(r)=\emptyset$. Since $\Delta_\phi=\emptyset$, $Dvalid$ trivially holds for all words. Further, $TValid$ checks the region template $\theta(\psi) = \mathscr R(\emptyset)$. Since $\mathscr R(\emptyset)$ is a region template with a single region, $Def^{\Delta_\phi,\theta_\phi}(\phi)$ is either a singleton or $\emptyset$. In the former case, the ranker $D^{\Delta_\phi,\theta_\phi}(\psi)$ exactly matches the position corresponding to the $dPos$ position of $\psi$. In the latter case, $D^{\Delta_\phi,\theta_\phi}(\psi)=\undef$. Hence part(ii) of the lemma is verified for the base case. 
\item Assume that $Iform(\psi)=\{\zeta_i\}$ is non-empty and the lemma holds for every $\zeta_i$ i.e., For every $\zeta_i=\fut/\past\phi_i$, Part(i) of the lemma holds for the restrictions $\Delta_{\phi_i},\theta_{\phi_i}$ and Part(ii) of the lemma holds for $\zeta_i$. We shall prove that the lemma holds for $\psi=\mathscr B(\phi)$.

Firstly, from the correctness of the construction of rankers for $\zeta_i$, we may verify the correctness of $Dvalid(\Delta_\phi)$ and $Tvalid(\theta_\phi)$. (Hence Part(i)). Further, from Lemma \ref{lem:def}, we know that $Def^{\Delta_\phi,\theta_\phi}(\phi)$ exactly marks the positions (in terms of regions and letter-occurrences within them) where $\phi$ holds. By observing the construction of rankers, we can infer that the ranker $D^{\Delta_\phi,\theta_\phi}(\psi)$ exactly matches the position corresponding to the $dPos$ position of $\psi$ (hence Part(ii)).
\end{itemize} 
\end{proof}

\paragraph{Constructing the formula $Trans(\Phi)$\\}
We may now give the language equivalent \tlxy\/ formula for the \utl\/ formula $\Phi$. Let $\Phi=\mathscr B(\{\psi_i,a_j\})$ where $\psi_i$ are immediate modal subformulas (which are only of the form $\fut\phi$ at the top level) and $a_j$ are atomic formulas. Then from the correctness of the validity formulas of the parameters $\Delta,\theta$ and rankers for the modal subformulas (Lemma \ref{lem:para}) we have
\[
  Trans(\Phi) ~=~ \bigvee\limits_{\Delta,\theta} ~[Dvalid(\Delta) ~\land~ Tvalid(\theta) ~\land~ \mathscr B(D^{\Delta,\theta}(\psi_i),a_j)].
\]

\paragraph{Complexity\\}
Consider a \utl\/ formula $\Phi$ of length $n$. Let $s$ be the size of its alphabet. The number of modal subformulas of $\Phi$ is $\mathcal O(n)$. For a given set of parameters $\Delta,\theta$,
\begin{itemize}
\item For each $\psi\in Sform(\Phi)$ the ranker $D^{\Delta,\theta}(\psi)$ is of size $\mathcal O(n)$.
\item Hence $Dvalid(\Delta)$ is of size $\mathcal O(n)$.
\item For each $\psi$, the size of $RValid(\psi,\theta)$ is $\mathcal O(n^3)$, and size of $AValid(\psi)$ and $BValid(\psi)$ is $\mathcal O(sn^2)$ 
\item $Tvalid$ checks the region template for each $\psi$. Hence the size of $Tvalid(\theta)$ is $\mathcal O(sn^4)$
\end{itemize}
Since the number of possible $\Delta$ and $\theta$ are exponential in $n$, $Trans(\Phi)$ is an $\mathcal O(2^n)$ disjunction of formulas whose size is bounded by $\mathcal O(sn^4)$.\\
\\
Time Complexity: For a given $\Delta,\theta$, the time taken to compute $Def^{\Delta,\theta}(\phi)$ for each $\phi$, is proportional to the number of regions and the size of $\phi$, i.e. $\mathcal O(n^2)$. Hence, the total time required to compute $Def$ for all subformulas is $\mathcal O(n^3)$. Further, the time required to compute the rankers for each modal subformula and the validity-checking formulas for $\Delta$ and $\theta$ is proportional to its size, which is polynomial in $n$. Hence we can conclude that the time taken to compute each disjunct of $Trans(\Phi)$ is also polynomial in $n$. 

\begin{theorem}
Satisfiability of \utl\/ formulas is decidable with NP-complete complexity.
\end{theorem}
\begin{proof}
For an input $\utl$ formula of size $n$, our reduction gives us a language equivalent $\tlxy$
formula of the form $\bigvee\limits_{i\in \{1\cdots k\}} \phi_i$ where $k$ is exponential in $n$ and each disjunct $\phi_i$ has a size polynomial in $n$ (assuming alphabet size to be a constant). 
From Proposition \ref{prop:unireg}, we know that the set of possible parameters $\Delta,\theta$ partitions $\Sigma^+$ into equivalence classes such that each equivalence class is characterized by the parameter to which the words in that class conform to. By non-deterministically guessing parameters $\Delta$ and $\theta$, a single disjunct $\phi_i$ may be constructed in time polynomial in $n$. By checking the satisfiability (which is in \np) of the resulting \tlxy\/ formula, we may check the satisfiability of the \utl\/ formula in \np\/ time. \np-hardness may be inferred from \np-hardness of propositional logic.
\end{proof}

The above construction results in a language equivalent $\potdfa$ whose number of states is exponential in $n$. However, every accepting path in the automaton has at most $O(n^4)$ progress (non-self looping) edges. 

\section{Recursive Logic \tlrecr}
\tlrecr\/ is the recursive extension of \tlxy\/ logic with deterministic modalities $X_{\psi}$ and $Y_{\psi}$ which are  parametrized by \tlrecr\/ sub-formulas $\psi$. The \tlrecr\/ formulas have a two-part syntax: subformulas may be $\phi$-type or $\psi$-type. They have the following syntax: \\
\[
\psi ~:=~ a ~\mid~ \phi ~\mid~ \psi\lor\psi ~\mid~ \neg\psi
\]
where $a\in \Sigma$ and $\phi$ is of the form\\
\[
\phi ~:= \top ~\mid~  SP\phi ~\mid~ EP\phi ~\mid~ X_\psi\phi ~\mid~ Y_\psi\phi \\
\]
Hence, the $\phi$-type formulas are \emph{recursive rankers} and the $X$ and $Y$ modalities are parametrized by $\psi$-type formulas which are boolean combinations of recursive rankers. On examining the above syntax representation, we may make the following key observations:
\begin{itemize}
\item The recursive rankers ($\phi$-type formulas) do not have $a$ as atomic subformulas. \footnote{It can be shown that allowing $a$ as an atomic subformula of a $\phi$-type formula increases the expressive power of the logic}.
\item Every $\psi$-type formula is a boolean combination of recursive rankers and atomic formulas.
\item The logic \tlrecr\/ is a deterministic logic and hence the subformulas satisfy the property of Unique Parsing. The unique position at which a subformula $n$ is evaluated in a given word $w$ is denoted by $\pos(n)$. 
\end{itemize}

The semantics of the recursive modalities of \tlrecr\/ formulas is as follows:\\
\begin{tabular}{rcl}
 $w,i\models X_{\phi_1}\phi_2$ & iff & $\exists j>i ~.~ w,j\models \phi_1\land w,j\models \phi_2$ and
     $\forall i<k<j ~.~ w,k\not\models \phi_1$\\
 $w,i\models Y_{\phi_1}\phi_2$ & iff & $\exists j<i ~.~ w,j\models \phi_1\land w,j\models \phi_2$ and
     $\forall j<k<i ~.~ w,k\not\models \phi_1$\\
\end{tabular}

\begin{example}
Consider the \tlrecr\/ formula $\phi =X_{\psi_1}Y_{\psi_2}\top$ where $\psi_1= a\land Y_b\top\land X_c\top$ and $\psi_2= X_c \habar{b}$. When we evaluate $\phi$ over the word $w= ccaccbccabbcacc$, $\pos(\phi)=1$. The first position in the word where $\psi_1$ holds is 9 hence $\pos(Y_{\psi_2}\top)=9$. Finally, the last position before 9 where $\psi_2$ holds is 4. Hence $w\in\mathcal L(\phi)$.
\end{example}

For a \tlrecr\/ formula $\psi$, the \textit{recursion level} of any subformula of $\psi$ may be defined inductively as follows: $\mathit{rlevel}(\psi)=0$. If $\phi=X_{\phi_1}\phi_2$ or $Y_{\phi_1}\phi_2$, then $\mathit{rlevel}(\phi_1)=\mathit{rlevel}(\phi)+1$ and
$\mathit{rlevel}(\phi_2)=\mathit{rlevel}(\phi)$. For all other operators, the recursion level remains unchanged. The \emph{recursion level} of a formula is the maximum recursion depth of its subformulas.

A key property of recursive rankers is \textit{convexity}. This is stated in the following lemma, and its proof is similar to that of Lemma \ref{lem:convex1}.
\begin{lemma}[Convexity]
\label{lem:convex2}
For any recursive ranker formula $\phi$, and any word $w\in\Sigma^+$, if there exist $i,j\in dom(w)$ such that $i<j$ and $w,i\models\phi$ and $w,j\models\phi$, then $\forall i<k<j$, we have $w,k\models\phi$. 
\end{lemma}

\subsection{\utl\/ to \tlrecr}
Consider a \utl\/ formula $\psi$ in \pnf: $\psi = a\land \land_i(\fut\alpha_i) \land \land_j(\past\beta_j) \land \land_k(\neg\fut\gamma_k) \land \land_l(\neg\past\delta_l)$. We construct the \tlrecr\/ formulas $\mathit{TransX}(\psi)$ and $\mathit{TransY}(\psi)$ such that the following lemma is satisfied.
\begin{lemma}
 If $\psi$ is a \utl\/ formula, then there exists a \tlrecr\/ formula $\mathit{Trans}(\psi)$ such that $\forall w\in \Sigma^+$ and $i\in dom(w)$, $w,i\models \psi$ iff $w,i\models \mathit{Trans}(\psi)$. Moreover, the size of $\mathit{Trans}(\psi)$ is linear in the size of $\psi$, and the modal depth of $psi$ is equal to the recursion depth of$\mathit{Trans}(\psi)$.
\end{lemma}
\begin{proof}
We now give the construction of $\mathit{Trans}(\psi)$, by structural induction on $\psi$. The correctness of the conversion is directly evident from the semantics of the two logics.
\begin{itemize}
 \item $\mathit{Trans}(a) = a$
 \item $\mathit{Trans}(\psi_1\lor\psi_2) =Trans(\psi_1)\lor Trans(\psi_2)$
 \item $\mathit{Trans}(\neg\psi) = \neg\mathit{Trans}(\psi)$
 \item $\mathit{Trans}(\fut(\psi)) = X_{\mathit{Trans}(\psi)}\top$
 \item $\mathit{Trans}(\past(\psi)) = Y_{\mathit{Trans}(\psi)}\top$
\end{itemize}
\end{proof}

\subsection{Reducing \tlrecr\/ to \utl}
For any \tlrecr\/ formula  $\psi$, we shall give a bottom-up inductive construction of  a \utl\/ formula $\at(\psi)$ such that the theorem below is satisfied.
\begin{theorem}
\label{theo:tlrecr}
For any $\psi \in \tlrecr$, we can construct \utl\/ formulas $\at(\psi)$ such that $\forall w\in\Sigma^+$, \\ 
$w,i\models \at(\psi)$ iff $w,i \models \psi$.
\end{theorem}
\begin{proof}
The proof is by induction on the structure of $\psi$ (and $\phi$). 
Define $\at(a) = a$, $\at(\top)=\top$ and $\at(\mathcal B(\phi_1, \ldots \phi_m)) = {\mathcal B(\at(\phi_1), \ldots \at(\phi_m))}$. It is easy to see that 
$w,j \models \at(\mathcal B(\phi_1, \ldots \phi_m))$ iff $w,j \models  \mathcal B(\phi_1, \ldots \phi_m)$. Now, we give and prove the reduction for temporal operators. 
\[
\begin{array}{lll}
\at(X_{\psi_1}(\phi_2)) & = &  ~~\fut[\at(\psi_1)\land \at(\phi_2)]  ~~\land \\
											&		 &  \neg\fut[\at(\psi_1)\land \neg 
											                     \at(\phi_2) \land F \at(\phi_2))] \\
\at(Y_{\psi_1}(\phi_2)) & = &  ~~\past[\at(\psi_1)\land \at(\phi_2)]  ~~\land \\
											&		 &  \neg\past[\at(\psi_1)\land 
											                     \neg \at(\phi_2) \land \past \at(\phi_2))] 		
\end{array}
\]
Consider the case $\phi=X_{\psi_1}(\phi_2)$. The other case is similar and omitted.
As $\phi=X_{\psi_1}(\phi_2)$ is a recursive ranker formula, the convexity property holds for $\phi$ and $\phi_2$ (but not always for $\psi_1$). This is depicted in the figure \ref{fig:tlrecr}. Using convexity, from the figure, the following property is evident: \\
$w,i  \models \phi \fif $ \\
$ \exists j >i \st w,j \models \psi_1 \land \phi_2 
~\mbox{and}~ \not \exists j > i \st w,j \models \psi_1 \land  \neg \phi_2 \land  \exists  k > j \st w,k \models \phi_2$ \\ 
$\fif w,i \models F(\phi_2 \land \psi_1) \land \neg F(\psi_i \land \neg \phi_2 
  \land F (\phi_2))$
\end{proof}

\begin{figure}
\begin{tikzpicture}
\draw (5,4.6) node{$\phi = X_{\psi_1}\phi_2$};
\draw (3,4.5) node{[} -- (7,4.5) node{]};
\draw (0,4) node{$w$}; \draw (0.2,4) node{l}-- (1,4) node{l} -- (2,4) node{l}--(3,4) node{l}--(4,4) node{l}--(5,4) node{l}--(6,4) node{l}--(7,4) node{l}--(8,4) node{l}--(9,4) node{l}--(10,4) node{l};
\draw (1,3.5) node{$\psi_1$}; \draw (3,3.5) node{$\psi_1$}; \draw (5,3.5) node{$\psi_1$}; \draw (6,3.5) node{$\psi_1$}; \draw (8,3.5) node{$\psi_1$}; \draw(10,3.5) node{$\psi_1$};
\draw (5,3.1) node{[}--(8,3.1) node{]};
\draw (6.5,3) node{$\phi_2$};

\oomit{
\draw (6,1.6) node{$\phi = Y_{\psi_1}\phi_2$};
\draw (4,1.5) node{[} -- (8,1.5) node{]};
\draw (0,1) node{$w$}; \draw (0.2,1) node{l}-- (1,1) node{l} -- (2,1) node{l}--(3,1) node{l}--(4,1) node{l}--(5,1) node{l}--(6,1) node{l}--(7,1) node{l}--(8,1) node{l}--(9,1) node{l}--(10,1) node{l};
\draw (1,0.5) node{$\psi_1$}; \draw (3,0.5) node{$\psi_1$}; \draw (5,0.5) node{$\psi_1$}; \draw (6,0.5) node{$\psi_1$}; \draw (8,0.5) node{$\psi_1$};
\draw (3,0.1) node{[}--(6,0.1) node{]};
\draw (4.5,0) node{$\phi_2$};
}
\end{tikzpicture}
\caption{Depicting convexity of recursive ranker $\phi=X_{\psi_1}\phi_2$}
\label{fig:tlrecr}
\end{figure}
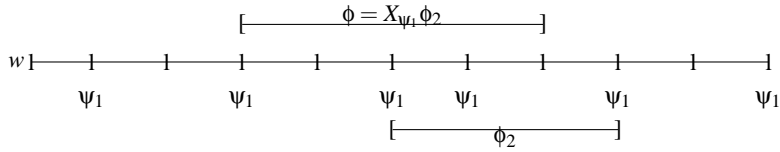

\paragraph{Complexity\\}
Consider a \tlrecr\/ formula $\psi$ of length $s$. We shall analyse the size of the language-equivalent \utl\/ formula. 
From the above construction, we can see that the modal DAG size of the resulting \utl\/ formula is linear in $s$ and hence its modal depth is also linear in $s$. 

Since the translation from \utl\/ to \potdfa\/ gives an \np-complete satisfiability procedure for \utl\/ formulas, the translation from \tlrecr\/ to \utl\/ gives an \np-complete satisfiability for \tlrecr\/ also.

\section{Discussion}
The motivation behind this study has been to use the various characterizations to help us in analyzing and answering some fundamental questions pertaining to this language class. Logic-automata transformations are important. They not only have practical applications in the form of model-checking, but also give more insight to the structure within the language class and its properties. Moreover, effective translations between various logics and automata allow us to calculate size-bounds, succinctness gaps and decision complexities. 

This study of unambiguous languages has also been extended to the language of factors (see \cite{LPS10}) and to timed words (see \cite{PS10}).

\bibliography{mybib}

\end{document}

%% file: macros.tex
\newcommand{\fif}{\rmiff}
\newcommand{\rmall}{\mathrel{\mbox{for all}~}}
\newcommand{\rmsome}{\mathrel{\mbox{for some}~}}
\newcommand{\rmand}{\mathbin{~\mbox{and}~}}
\newcommand{\rmiff}{\mathbin{~\mbox{iff}~}}

\newcommand{\name}{\textsf}
\newcommand{\class}{\textsc}

\newcommand{\defn}{\mathrel{\mbox{$~\stackrel{\rm def}{=}~$}}}

\newcommand{\undef}{\textbf{u}}

\newcommand{\opr}{\mbox{$\mathit{Opr}$}}
\newcommand{\sub}{\mbox{$\mathit{Subf}$}}
\newcommand{\intv}{\mbox{$\mathit{Intv}_w$}}

\newcommand{\pos}{\mbox{$\mathit{Pos}_w$}}
\newcommand{\lpos}{\mbox{$\ell\pos$}}
\newcommand{\cpos}{\mbox{$c\pos$}}

\newcommand{\fut}{\textsf{F}}
\newcommand{\past}{\textsf{P}}

\newcommand{\first}[1]{\mbox{$F_{#1}$}}
\newcommand{\firstp}[1]{\mbox{$F^+_{#1}$}}

\newcommand{\last}[1]{\mbox{$L_{#1}$}}

\newcommand{\lastm}[1]{\mbox{$L^-_{#1}$}}
\newcommand{\firsta}{\first{a}}

\newcommand{\lasta}{\last{a}}

\newcommand{\uitl}{\mbox{$\mathit{UITL}$}}

\newcommand{\fo}{\mbox{$\mathit{FO}$}}

\newcommand{\fotwo}{\mbox{$\fo^2$\/}}
\newcommand{\fotwoless}{\mbox{$\fotwo[<]$\/}}

\newcommand{\POTDFAO}{\mbox{$\mathit{po2dfa}$}}
\newcommand{\POTDFA}{\mbox{$\mathit{po2dfa}$}}
\newcommand{\fscan}[2]{#1 \stackrel{#2}{\rightarrow}}
\newcommand{\bscan}[2]{#1 \stackrel{#2}{\leftarrow}}
\newcommand{\funit}[1]{1 \stackrel{#1}{\rightarrow}}
\newcommand{\bunit}[1]{1 \stackrel{#1}{\leftarrow}}

\newcommand{\oomit}[1]{}

\newcommand{\st}{.~}
\newcommand{\pti}{\mathit{pt}}

\newcommand{\lcceil}{\hbox{$\lceil\!\lceil$}}
\newcommand{\rcceil}{\hbox{$\rceil\!\rceil$}}
\newcommand{\aloo}[1]{\lceil #1 \rceil}
\newcommand{\alcc}[1]{\lcceil #1 \rcceil}
\newcommand{\aloc}[1]{\lceil #1 \rcceil}
\newcommand{\alco}[1]{\lcceil #1 \rceil}

\newcommand{\fsearch}[1]{\stackrel{#1}{\rightarrow}}
\newcommand{\bsearch}[1]{\stackrel{#1}{\leftarrow}}
\newcommand{\fstep}{\stackrel{1}{\rightarrow}}
\newcommand{\bstep}{\stackrel{1}{\leftarrow}}

\newcommand{\rend}{\triangleleft}
\newcommand{\lend}{\triangleright}

\newcommand{\smartqed}

\newcommand{\succr}{\oplus}
\newcommand{\predr}{\ominus}

\newcommand{\ETE}{ETE}

\newcommand{\autm}{\mathcal A}

\newcommand{\potdfa}{\POTDFA}

\newcommand{\stl}{\mbox{$\mathit{Ranker ~Formula}$}}
\newcommand{\stls}{\mbox{$\mathit{Ranker ~Formulas}$}}

\newcommand{\utl}{\mbox{$\mathit{TL[F,P]}$}}
\newcommand{\pnf}{\textit{normal form}}

\newcommand{\gabar}[1]{G_{\overline{#1}}}
\newcommand{\habar}[1]{H_{\overline{#1}}}

\newcommand{\ete}{\mbox{$\mathit{ETE}$}}

\newcommand{\weakx}[1]{\widetilde{X}_{#1}} 
\newcommand{\weaky}[1]{\widetilde{Y}_{#1}}
\newcommand{\xunit}{X_1}
\newcommand{\yunit}{Y_1}
\newcommand{\pfless}{\mathcal P^{<}}
\newcommand{\pfleq}{\mathcal P^{\leq}}
\newcommand{\pfgreat}{\mathcal P^{>}}
\newcommand{\pfgeq}{\mathcal P^{\geq}}

\newcommand{\tlxy}{\mbox{$\mathit{TL[X_a,Y_a]}$}}

\newcommand{\conp}{\mbox{$\class{CoNP}$}}
\newcommand{\np}{\mbox{$\class{NP}$}}

\newcommand{\tlrecr}{\mbox{$\mathit{TL}^+[X_\phi,Y_\phi]$}}

\newcommand{\at}{\mbox{$\mathit{At}$}}
\newcommand{\uitlpm}{\mbox{$\mathit{UITL^{\pm}}$}}

\newcommand{\tab}{\hspace*{1cm}}
\newcommand{\leafn}{\mathit{Leaf}}

\newcommand{\detu}{\mbox{$\mathit{\widetilde{U}}$}}
\newcommand{\dets}{\mbox{$\mathit{\widetilde{S}}$}}
\newcommand{\duds}{\mbox{$\mathit{TL[\detu,\dets]}$}}
\newcommand{\succp}{\overline{\oplus}}
\newcommand{\predm}{\overline{\ominus}}